\documentclass[letterpaper,11pt]{article}
\usepackage{fullpage}
\usepackage{amsmath,amsthm,amsfonts,dsfont}
\usepackage{amssymb,latexsym,graphicx}
\usepackage{mathtools}
\usepackage{hyperref}
\usepackage{xcolor}
\hypersetup{
	colorlinks,
	linkcolor={red!75!black},
	citecolor={blue!75!black},
	urlcolor={blue!75!black}
}

\newtheorem{theorem}{Theorem}[section]

\newtheorem{lemma}[theorem]{Lemma}

\newtheorem{claim}[theorem]{Claim}

\newtheorem{corollary}[theorem]{Corollary}
\newtheorem{definition}[theorem]{Definition}

\newcommand{\R}{\ensuremath{\mathbb{R}}}
\newcommand{\Z}{\ensuremath{\mathbb{Z}}}

\newcommand{\lat}{\mathcal{L}}

\newcommand{\eps}{\varepsilon} 
\renewcommand{\epsilon}{\varepsilon}
\newcommand{\poly}{\mathrm{poly}}

\newcommand{\rank}{\mathrm{rank}}

\DeclareMathOperator*{\expect}{\mathbb{E}}
\DeclareMathOperator{\dist}{dist}

\renewcommand{\vec}[1]{\ensuremath{\boldsymbol{#1}}}
\newcommand{\basis}{\ensuremath{\mathbf{B}}}

\newcommand{\problem}[1]{\ensuremath{\mathrm{#1}} }

\DeclarePairedDelimiter\inner{\langle}{\rangle}
\DeclarePairedDelimiter\ceil{\lceil}{\rceil}
\newcommand{\CVP}{\problem{CVP}}
\newcommand{\SVP}{\problem{SVP}}
\newcommand{\SAT}{\problem{SAT}}
\newcommand{\ESC}{\problem{ExactSetCover}}
\newcommand{\cL}{{\mathcal L}}

\newcommand{\cC}{{\mathcal C}}

\newcommand{\tC}{\widetilde {C}}

 \newif\iffull
 \fulltrue 

 \newcommand{\full}[2]{\iffull#1\else#2\fi} %

\begin{document}
	
	\title{(Gap/S)ETH Hardness of SVP}
	\author{Divesh Aggarwal\thanks{Supported by the Singapore Ministry of Education and the National Research Foundation, also through the Tier 3 Grant ``Random numbers from quantum processes" MOE2012-T3-1-009.} \\
		Centre for Quantum Technologies, NUS\\
		\texttt{divesh.aggarwal@gmail.com}
		\and 
		Noah Stephens-Davidowitz\thanks{Supported by the Simons Collaboration on Algorithms and Geometry.}\\
		Princeton University\\
		\texttt{noahsd@gmail.com}
			 }
	\date{}
	\maketitle
	
	\begin{abstract}
		We prove the following quantitative hardness results for the Shortest Vector Problem in the $\ell_p$ norm ($\SVP_p$), where $n$ is the rank of the input lattice. 
		\begin{enumerate}
			\item For ``almost all'' $p > p_0 \approx 2.1397$, there no $2^{n/C_p}$-time algorithm for $\SVP_p$ for some explicit constant $C_p > 0$ unless the (randomized) Strong Exponential Time Hypothesis (SETH) is false.
			\item For any $p > 2$, there is no $2^{o(n)}$-time algorithm for $\SVP_p$ unless the (randomized) Gap-Exponential Time Hypothesis (Gap-ETH) is false. Furthermore, for each $p > 2$, there exists a constant $\gamma_p > 1$ such that the same result holds even for $\gamma_p$-approximate $\SVP_p$.
			\item There is no $2^{o(n)}$-time algorithm for $\SVP_p$ for any $1 \leq p \leq 2$ unless either (1) (non-uniform) Gap-ETH is false; or (2) there is no family of lattices with exponential kissing  number in the $\ell_2$ norm. Furthermore, for each $1 \leq p \leq 2$, there exists a constant $\gamma_p > 1$ such that the same result holds even for $\gamma_p$-approximate $\SVP_p$.
		\end{enumerate}
	\end{abstract}
	
	\setcounter{page}{0}
	\thispagestyle{empty}
	\newpage
	\setcounter{page}{1}
	
	\section{Introduction}
	\label{sec:intro}

A lattice $\lat$ is the set of all integer combinations of
linearly independent basis vectors $\vec{b}_1,\dots,\vec{b}_n \in \R^d$,
\[
\lat = \lat(\vec{b}_1,\ldots, \vec{b}_n) := \Big\{ \sum_{i=1}^n z_i \vec{b}_i \ : \ z_i \in \Z \Big\}
\; .
\] We call $n$ the \emph{rank} of the lattice $\lat$ and $d$ the \emph{dimension} or the \emph{ambient dimension} of the lattice $\cL$.

The Shortest Vector Problem ($\SVP$) takes as input a basis for a lattice $\lat \subset \R^d$ and $r > 0$ and asks us to decide whether the shortest non-zero vector in $\lat$ has length at most $r$. Typically, we define length in terms of the $\ell_p$ norm for some $1 \leq p \leq \infty$, defined as
\[
\|\vec{x}\|_p := (|x_1|^p + |x_2|^p + \cdots + |x_d|^p)^{1/p}
\]
for finite $p$ and 
\[
\|\vec{x}\|_\infty := \max |x_i|
\; .
\]
In particular, the $\ell_2$ norm is the familiar Euclidean norm, and it is the most interesting case from our perspective.
We write $\SVP_p$  for $\SVP$ in the $\ell_p$ norm (and just $\SVP$ when we do not wish to specify a norm).

Starting with the breakthrough work of Lenstra, Lenstra, and Lov{\'a}sz in 1982~\cite{LLL82}, algorithms for solving $\SVP$ in both its exact and approximate forms have found innumerable applications, including
factoring polynomials over the rationals~\cite{LLL82}, integer programming~\cite{Lenstra83,Kannan87,DPV11}, cryptanalysis~\cite{Shamir84,Odl90,JS98,NS01}, etc. More recently, many cryptographic primitives have been constructed whose security is based on the (worst-case) hardness of $\SVP$ or closely related lattice problems \cite{Ajtai96,oded05,GPV08,Pei10,chris_survey}. Such lattice-based cryptographic constructions are likely to be used on massive scales (e.g., as part of the TLS protocol) in the not-too-distant future~\cite{new_hope,frodo,NIST_quantum}. 

Most of the above applications rely on approximate variants of $\SVP$ with rather large approximation factors (e.g., the relevant approximation factors are polynomial in $n$ for most cryptographic constructions). However, the best known algorithms for the approximate variant of $\SVP$ use an algorithm for exact $\SVP_2$ over lower-rank lattices as a subroutine~\cite{Schnorr87,GN08,MW16}. So, the complexity of the exact problem is of particular interest. We briefly discuss some of what is known below.

\paragraph{Algorithms for $\SVP$. } 
Most of the asymptotically fastest known algorithms for $\SVP$ are variants of the celebrated randomized sieving algorithm due to Ajtai, Kumar, and Sivakumar~\cite{AKS01}, which solved $\SVP_p$ in $2^{O(d)}$ time for $p = 2$ and $p=\infty$.
This was extended to all $\ell_p$ norms~\cite{BN09}, then to $\SVP$ in all norms~\cite{AJ08}, and then even to ``norms'' whose unit balls are not necessarily symmetric~\cite{DPV11}. These $2^{O(d)}$-time algorithms that work in all norms in particular imply $2^{O(n)} \cdot \poly(d)$-time algorithms for $\SVP_p$, by taking the ambient space to be the span of the lattice. We are therefore primarily interested in the running time of these algorithms as a function of the rank $n$. (Notice that, in the $\ell_2$ norm, we can always assume that $n=d$.)

In the special case of $p = 2$, quite a bit of work has gone into improving the constant in the exponent in these $2^{O(n)}$-time algorithms~\cite{NguyenVidick08,PS09,MV10,LWXZ11}. The current fastest known algorithm for $\SVP_2$ runs in $2^{n + o(n)}$ time~\cite{ADRS15,AS17}. But, this is unlikely to be the end of the story. Indeed, there is also a $2^{n/2 + o(n)}$-time algorithm that approximates $\SVP_2$ up to a small constant factor,\footnote{%
	Unlike all other algorithms mentioned here, this $2^{n/2 + o(n)}$-time algorithm does not actually find a short vector; it only outputs a length. In the exact case, these two problems are known to be equivalent under an efficient rank-preserving reduction~\cite{MicciancioBook}, but this is not known to be true in the approximate case.%
	}
 and there is some reason to believe that this algorithm can be modified to solve the exact problem~\cite{ADRS15,AS17}. Further complicating the situation, there exist even faster ``heuristic algorithms,'' whose correctness has not been proven but can be shown under certain heuristic assumptions~\cite{NguyenVidick08,WLTB11,Laarhoven2015}. The fastest such heuristic algorithm runs in time $(3/2)^{n/2 + o(n)} \approx 2^{0.29 n}$~\cite{BDGL16}.
 
\paragraph{Hardness of $\SVP$. } Van Emde Boaz first asked whether $\SVP_p$ was NP-hard in 1981, and he proved NP-hardness in the special case when $p = \infty$~\cite{Boas81}. Despite much work, his question went unanswered until 1998, when Ajtai showed NP-hardness of $\SVP_p$ for all $p$~\cite{Ajtai-SVP-hard}. A series of works by Cai and Nerurkar~\cite{CN98}, Micciancio~\cite{Mic01svp}, Khot~\cite{Khot05svp}, and Haviv and Regev~\cite{HRsvp} simplified the reduction and showed hardness of the approximate version of $\SVP_p$. We now know that $\SVP_p$ is NP-hard to approximate to within any constant factor and hard to approximate to within approximation factors as large as $n^{c/\log \log n}$ for some constant $c > 0$ under reasonable complexity-theoretic assumptions.\footnote{%
	All of these reductions for finite $p$ are randomized, as are ours. An unconditional deterministic reduction would be a major breakthrough. See~\cite{Mic01svp,Mic12} for more discussion and even a conditional deterministic reduction that relies on a certain number-theoretic assumption.%
} 

However, such hardness proofs tell us very little about the \emph{quantitative} or \emph{fine-grained} complexity of $\SVP_p$. E.g., does the fastest possible algorithm for $\SVP_2$ still run in time at least, say, $2^{n/5}$, or is there an algorithm that runs in time $2^{n/20}$ or even $2^{\sqrt{n}}$? The above hardness results cannot distinguish between these cases, but we certainly need to be confident in our answers to such questions if we plan to base the security of widespread cryptosystems on these answers. Indeed, most proposed instantiations of lattice-based cryptosystems (i.e., proposed cryptosystems that specify a key size) can essentially be broken by solving $\SVP_2$ with, say, $n \ll 600$ or $\SVP_p$ for any $p$ with $n \ll 1500$. So, if we discovered an algorithm running in time, say, $2^{n/20}$-time for $\SVP_2$ or $2^{n/50}$ or $2^{n/\log^2 n}$ for $\SVP_p$, then these schemes would be broken in practice. And, given the large number of recent algorithmic advances, one might (reasonably?) worry that such algorithms will be found. We would therefore very much like to rule out this possibility!

To rule out such algorithms, we typically rely on a fine-grained complexity-theoretic hypothesis---such as the Strong Exponential Time Hypothesis (SETH\full{, see Section~\ref{sec:fine-grained_prelims}}{}) or the Exponential Time Hypothesis (ETH). To that end, Bennett, Golovnev, and Stephens-Davidowitz recently showed quantitative hardness results for the Closest Vector Problem in $\ell_p$ norms ($\CVP_p$)~\cite{BGS17}, which is a close relative of $\SVP_p$ that is known to be at least as hard (so that this was a necessary first step towards proving similar results for $\SVP_p$). In particular, assuming SETH,~\cite{BGS17} showed that there is no $2^{(1-\eps) n}$-time algorithm for $\CVP_p$ or $\SVP_\infty$ for any $\eps > 0$ and ``almost all'' $1 \leq p \leq \infty$ (\emph{not} including $p = 2$). Under ETH,~\cite{BGS17} showed that there is no $2^{o(n)}$-time algorithm for $\CVP_p$ for any $1 \leq p \leq \infty$. We prove similar results for $\SVP_p$ for $p > 2$ (and a conditional result for $1 \leq p\le 2$ that holds if there exists a family of lattices satisfying certain geometric conditions).

\subsection{Our results}

We now present our results, which are also summarized in Table~\ref{tab:complexity_summary}. 

\newcommand{\boldandblueexp}{\color{blue!70!black} $\mathbf{2^{\Omega(n)}}$}

\begin{table}
	\begin{center}
		\begin{tabular}{| c | c | c |  c | c | c|}
			\hline
			 & Upper Bound & \multicolumn{3}{|c|}{Lower Bounds} &Notes\\
			\hline
			&  & SETH & Gap-ETH & \begin{tabular}{c}Gap-ETH +\\Kissing Number\end{tabular} & \\
			\hline
			$p_0 < p < \infty$  & $2^{O(n)}$& {\color{blue!70!black} $2^{n/C_p}$} &  \boldandblueexp & \boldandblueexp & See Fig.~\ref{fig:Cp}.  \\
			$2 < p \leq p_0$ & $2^{O(n)}$ & -- & \boldandblueexp & \boldandblueexp & $p_0 \approx 2.1397$  \\
			$p = 2$ & $2^{n}$ ($2^{0.3 n}$) & -- & -- & \boldandblueexp &  \\
			$1 \leq p < 2$ & $2^{O(n)}$ &-- &-- & \boldandblueexp &  \\
			$p = \infty$ & $2^{O(n)}$ & $\mathbf{2^n}$ & $\mathbf{2^{\Omega(n)}}$ & $\mathbf{2^{\Omega(n)}}$ & \cite{BGS17} \\
			\hline
		\end{tabular}
		\caption{\label{tab:complexity_summary} Summary of known fine-grained upper and lower bounds for $\SVP_p$ for various $p$ under various assumptions, with new results in {\color{blue!70!black}blue}. Lower bounds in {\bf bold} also apply for some constant approximation factor strictly greater than one. The one upper bound in parentheses is due to a heuristic algorithm. The SETH-based lower bound only applies for ``almost all'' $p > p_0$\full{, in the sense of Theorem~\ref{thm:SETH_hardness_centered_theta}}{ (as defined in the full version)}. We have suppressed low-order terms for simplicity.}
	\end{center}
\end{table}

\paragraph{SETH-hardness.} Our first main result essentially gives an explicit constant $C_p > 1$ for each $p > p_0 \approx 2.1397$ such that, under (randomized) SETH, there is no algorithm for $\SVP_p$ that runs in time better than $2^{n/C_p}$. The constants $p_0$ and $C_p$ do not have a closed form, but they are easily computable to high precision in practice.
(E.g., $p_0 = 2.13972134795007\ldots$, $C_3 = 
3.01717780317660\ldots$, and $C_5 = 1.3018669052709\ldots$.) We plot $C_p$ over a wide range of $p$ in Figure~\ref{fig:Cp}. Notice that $C_p$ is unbounded as $p$ approaches $p_0$, but it is a relatively small constant for, say, $p \gtrsim 3$.

We present this result informally here, as the actual statement is rather technical. In particular, because we use the theorem from~\cite{BGS17} that only applies to ``almost all'' $p$, our result also has this property. See \full{Theorem~\ref{thm:SETH_hardness_centered_theta}}{the full version} for the formal statement.

		\begin{figure}[ht]
			\begin{center}
				\includegraphics[width=0.4 \textwidth]{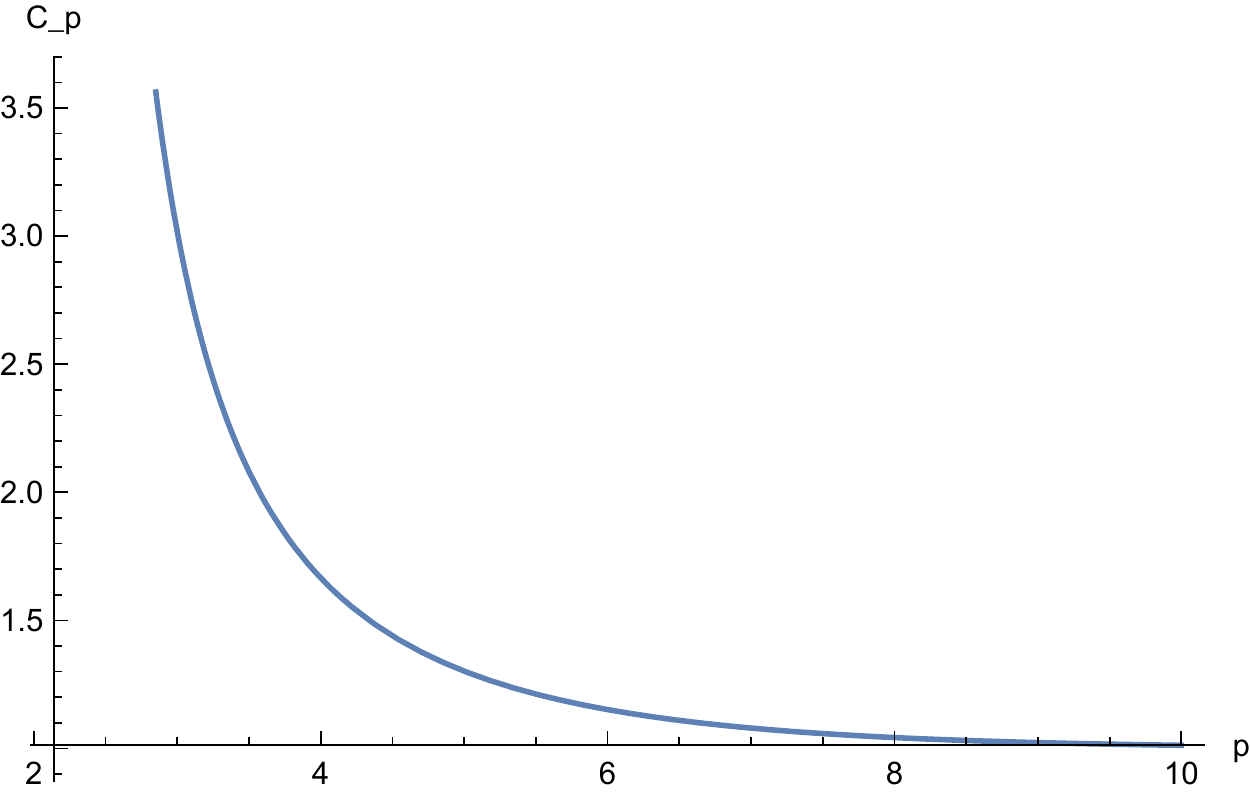}
				\qquad
				\includegraphics[width=0.4 \textwidth]{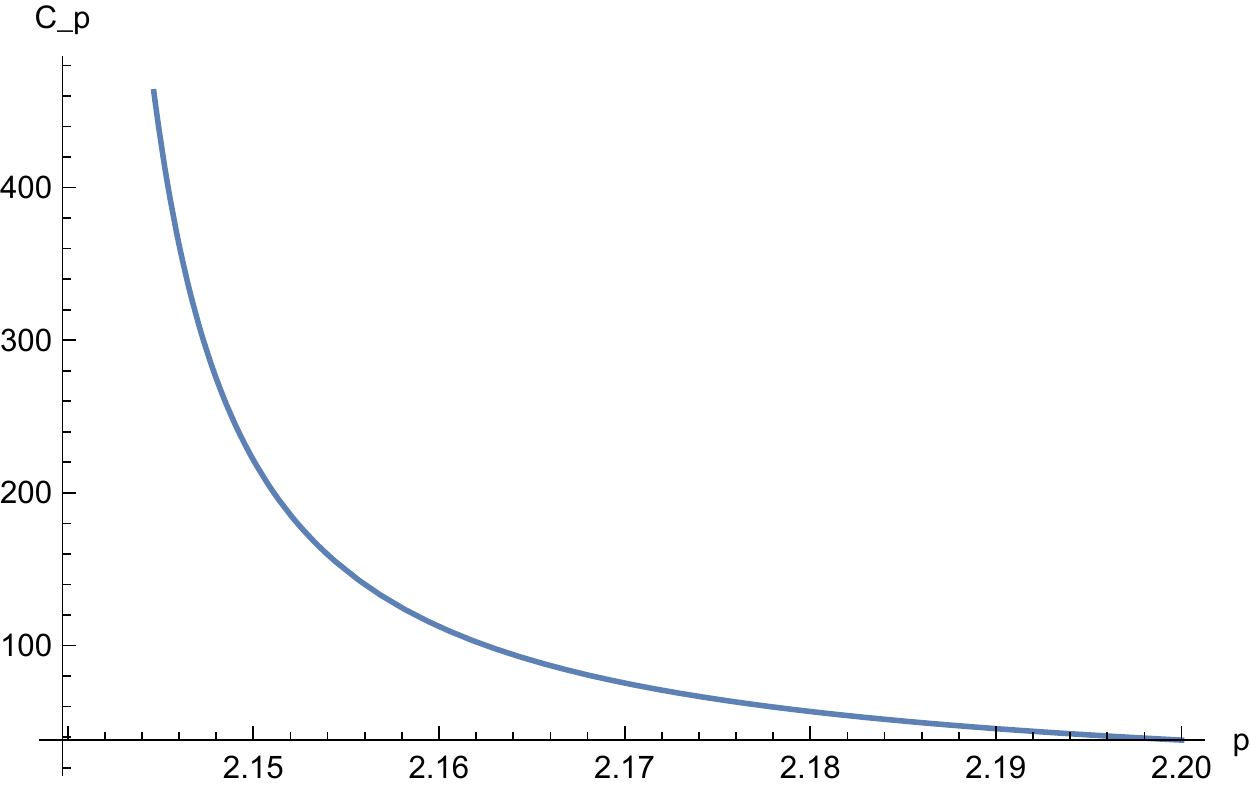}
				\caption{\label{fig:Cp} 
					The value $C_p$ for different values of $p > p_0$. In particular, up to the minor technical issues \full{in Theorem~\ref{thm:SETH_hardness_centered_theta}}{discussed in the full version}, there is no $2^{n/C_p}$-time algorithm for $\SVP_p$ unless SETH is false. The plot on the left shows $C_p$ over a wide range of $p$, while the plot on the right shows the behavior when $p$ is close to its minimal value $p_0 \approx 2.13972$.}
			\end{center}
		\end{figure}

\begin{theorem}[Informal]
	\label{thm:SETH_intro}
	For ``almost all'' $p > p_0 \approx 2.1397$ (including all odd integers $p \geq 3$), there is no $2^{n/C_p}$-time algorithm for $\SVP_p$ unless (randomized) SETH is false, where $C_p > 1$ is as in Figure~\ref{fig:Cp}. Furthermore, $C_p \to 1$ as $p \to \infty$.
\end{theorem}

To prove this theorem, we give a (randomized) reduction from the $\CVP_p$ instances created by the reduction of~\cite{BGS17} to $\SVP_p$ that only increases the rank of the lattice by a constant factor. As we describe in Section~\ref{sec:techniques}, our reduction is surprisingly simple. 
In particular, the key step in Khot's reduction~\cite{Khot05svp} uses a certain ``gadget'' consisting of a lattice $\lat^\dagger$, vector $\vec{t}^\dagger$, and distance $r^\dagger > 0$ to convert a provably hard $\CVP_p$ instance into an $\SVP_p$ instance. Our reduction is similar to Khot's reduction with the simple gadget given by $\lat^\dagger := \Z^{n^\dagger}$, $\vec{t}^\dagger := (1/2,\ldots, 1/2) \in \R^{n^\dagger}$, and $r^\dagger := n^{1/p}/2$. 

We note in passing that we actually do not need the full strength of SETH. We can instead rely on the analogous assumption for Max-$k$-SAT, which is potentially weaker. (We inherit this property directly from~\cite{BGS17}. See \full{Section~\ref{sec:gadget_Zn_all_halves}}{the full version}.)

\paragraph{Gap-ETH-hardness.} Our second main result is the Gap-ETH-hardness of $\SVP_p$ for all $p > 2$.\footnote{Gap-ETH is the assumption that there is no $2^{o(n)}$-time algorithm that distinguishes a satisfiable $3$-SAT formula from one in which at most a constant fraction of the clauses are simultaneously satisfiable.\full{ See Section~\ref{sec:fine-grained_prelims}.}{}} In fact, we prove this even for the problem of approximating $\SVP_p$ up to some fixed constant $\gamma_p > 1$ depending only on $p$ (and the approximation factor implicit in the Gap-ETH assumption). \full{See Corollary~\ref{cor:Gap3SAT_to_SVP}.}{}

\begin{theorem}[Informal]
	\label{thm:ETH_intro}
	For any $p > 2$, there is no $2^{o(n)}$-time algorithm for $\SVP_p$ unless (randomized) Gap-ETH is false. Furthermore, for each such $p$ there is a constant $\gamma_p > 1$ such that the same result holds even for $\gamma_p$-approximate $\SVP_p$.
\end{theorem}

Our reduction is again quite simple (though the proof of correctness is not). We follow Khot's reduction from approximate Exact Set Cover, and we again use the integer lattice as our gadget (with a different target).\footnote{We note that Khot claimed in Section 8 of~\cite{Khot05svp} that he had discovered a linear reduction from $\gamma'$-approximate $\CVP_p$ to $2^{1-3/p}$-approximate $\SVP_p$ for $p \geq 4$ and some unspecified constant $\gamma'$. However, it is not clear whether $\gamma'$ is a small enough constant to yield an alternate proof of Theorem~\ref{thm:ETH_intro} for $p \ge 4$. In particular, one would need to show Gap-ETH-hardness of $\gamma'$-approximate $\CVP_p$.} 

We note in passing that for this result (as well as Theorem~\ref{thm:kissing_intro} and Corollary~\ref{cor:GapETH_kissing_intro}), we actually rule out even $2^{o(d)}$-time algorithms. However, we focus on the rank $n$ instead of the dimension $d$ for simplicity.

\paragraph{Towards $p = 2$.} We are unable to extend either Theorem~\ref{thm:SETH_intro} or Theorem~\ref{thm:ETH_intro} to the important case when $p = 2$. Indeed, we cannot use the integer lattice as a gadget in the Euclidean norm. However, we do show that the existence of a certain type of lattice that is believed to exist would be sufficient to show (possibly non-uniform) Gap-ETH-hardness of $\SVP_2$. In particular, it would suffice to show the existence of any family of lattices with exponentially large kissing number. See Theorem~\ref{thm:kissing_gives_hardness} for the precise statement, which requires the existence of a structure that might be easier to construct (and see, e.g.,~\cite{Alon97,ConwaySloaneBook98} for discussion of the lattice kissing number).

\begin{theorem}[Informal]
	\label{thm:kissing_intro}
	There is no $2^{o(n)}$-time algorithm for $\SVP_2$ unless either (1) (non-uniform) Gap-ETH is false; or (2) the lattice kissing number is $2^{o(n)}$. Furthermore, there exists a constant $\gamma > 1$ such that the same result holds even for $\gamma$-approximate $\SVP_2$.
\end{theorem}

In fact, Regev and Rosen show that $\ell_2$ is in some sense the ``easiest norm''~\cite{RR06}. \full{(See Theorem~\ref{thm:embedding}.) }{}In particular, to show that $\SVP_p$ is Gap-ETH-hard for all $1 \leq p \leq 2$, it suffices to show it for $p = 2$. From this, we derive the following. (See \full{Corollary~\ref{cor:ellp_hard}}{the full version of this paper} for the formal statement.)

\begin{corollary}[Informal]
	\label{cor:GapETH_kissing_intro}
	There is no $2^{o(n)}$-time algorithm for $\SVP_p$ for any $1 \leq p \leq 2$ unless either (1) (non-uniform) Gap-ETH is false; or (2) the lattice kissing number is $2^{o(n)}$ (in the $\ell_2$ norm). Furthermore, for each $1 \leq p \leq 2$, there exists a constant $\gamma_p > 1$ such that the same result holds even for $\gamma_p$-approximate $\SVP_p$.
\end{corollary}

\subsection{Khot's reduction}
\label{sec:Khot}

Before we describe our own contribution, it will be useful to review Khot's elegant reduction from $\CVP_p$ to $\SVP_p$~\cite{Khot05svp}.
We do our best throughout this description to hide technicalities in an effort to focus on the high-level simplicity of Khot's reduction.\footnote{%
		Khot's primary motivation for his reduction was to prove hardness of approximating $\SVP_p$ to within any constant factor, by showing a reduction that is well-behaved under a certain tensor product. We are not interested in taking tensor products (since they produce lattices of superlinear rank), so we ignore this issue entirely.
	} (Since the hardness of $\SVP_p$ went unproven for many years, this simplicity is truly remarkable.)

First, some basic definitions and notation. For a lattice $\lat \subset \R^d$ and $1 \leq p \leq \infty$, we write
\[
\lambda_1^{(p)}(\lat) := \min_{\vec{y} \in \lat \setminus \{\vec0\}} \|\vec{y}\|_p
\]
for the length of the shortest non-zero vector in $\lat$ in the $\ell_p$ norm. For a target vector $\vec{t} \in \R^d$, we write
\[
\dist_p(\vec{t}, \lat) := \min_{\vec{y} \in \lat} \|\vec{y} - \vec{t}\|_p
\]
for the distance between $\vec{t}$ and $\lat$. For any radius $r > 0$, we write
\[
N_p(\lat, r, \vec{t}) := |\{ \vec{y} \in \lat \ : \ \|\vec{y} - \vec{t}\|_p \leq r \}|
\]
for the number of lattice vectors within distance $r$ of $\vec{t}$.

Recall that $\CVP_p$ is the problem that takes as input a lattice $\lat \subset \R^d$, target vector $\vec{t} \in \R^d$, and distance $r > 0$ and asks us to distinguish the YES case when $\dist_p(\vec{t}, \lat) \leq r$ from the NO case when $\dist_p(\vec{t}, \lat) > r$.
When talking about a particular $\CVP_p$ instance, we naturally call a lattice vector $\vec{y} \in \lat$ with $\|\vec{y} - \vec{t}\|_p \leq r$ a \emph{close vector}, and we notice that the number of close vectors is $N_p(\lat, r, \vec{t})$.

\paragraph{The naive reduction and sparsification.} The ``naive reduction'' from $\CVP_p$ to $\SVP_p$ simply takes a $\CVP_p$ instance consisting of a lattice $\lat \subset \R^d$ with basis $\basis \in \R^{d \times n}$, target $\vec{t} \in \R^d$, and distance $r > 0$ and constructs the $\SVP_p$ instance given by the basis of a lattice $\lat'$ of the form
\[
\basis' := 
\begin{pmatrix}
\basis &-\vec{t}\\
0 &s
\end{pmatrix}
\; ,
\]
where $s > 0$ is some parameter depending on the $\CVP_p$ instance. 
Notice that, if $\vec{y} \in \lat$ is a close vector (i.e., $\|\vec{y} - \vec{t}\| \leq r$), then $\|(\vec{y} - \vec{t}, s) \|_p^p \leq r^p + s^p$. Therefore, in the YES case when there exists a vector close to $\vec{t}$, we will have $\lambda_1^{(p)}(\lat') \leq r' := (r^p + s^p)^{1/p}$.

However, in the NO case there might still be non-zero vectors $\vec{y}' \in \lat' \setminus \{\vec0\}$ whose length is less than $r'$. These vectors must be of the form $\vec{y}' = (\vec{y} - z\vec{t}, zs)$ for some integer $z \neq 1$.
Let us for now only consider the case $z = 0$, in which case these vectors are in one-to-one correspondence with the non-zero vectors in $\lat$ of length less than $r'$. 
We naturally call these \emph{short vectors}.

Khot showed that a (randomized) reduction exists if we just assume that the number of close vectors in any YES case is significantly larger than the number of short vectors in any NO case. In particular, Khot showed that we can randomly ``sparsify'' the lattice $\lat'$ to obtain a sublattice $\lat''$ such that each of the short non-zero vectors in $\lat'$ stays in $\lat''$ with probability $1/q$ where $q \geq 2$ is some parameter that we can choose.
So, if we take $q$ to be significantly smaller than the number of close vectors in the YES case but significantly larger than the number of short vectors in the NO case, we can show that the resulting lattice will have $\lambda_1^{(p)}(\lat) \leq r'$ in the YES case but $\lambda_1^{(p)}(\lat) > r'$ in the NO case with high probability.

Unfortunately, the $\CVP_p$ instances produced by most hardness reductions typically have $2^{\Omega(n)}$ short vectors, and they might only have one close vector in the YES case. So, if we want this reduction to work, we will need some way to increase this ratio by an exponential factor.

\paragraph{Adding the gadget.} To increase the ratio of close vectors to short vectors, Khot uses a certain gadget that is itself a $\CVP_p$ instance $(\lat^\dagger, \vec{t}^\dagger, r^\dagger)$, where $\lat^\dagger \subset \R^{d^\dagger}$ is a lattice with basis $\basis^\dagger$, $\vec{t}^\dagger \in \R^{d^\dagger}$ is a target vector, and $r^\dagger > 0$ is some distance. He then takes the direct sum of the two instances. I.e., Khot considers the lattice
\[
\widehat{\lat} := \lat \oplus \lat^\dagger =  \{ (\vec{y}, \vec{y}^\dagger)  \ : \ \vec{y} \in \lat, \vec{y}^\dagger \in \lat^\dagger\} \subset \R^{d + d^\dagger}
\]
with basis
\[
\widehat{\basis} := \begin{pmatrix}
\basis &0 \\
0 & \basis^\dagger
\end{pmatrix}
\; ,
\]
the target $\widehat{\vec{t}} := (\vec{t}, \vec{t}^\dagger) \in \R^{d + d^\dagger}$,
and the distance $\widehat{r} := (r^p + (r^\dagger)^{p})^{1/p}$. We wish to apply the sparsification-based reduction described above to this new lattice. So, we proceed to make some observations about $\lat^\dagger$ to deduce some properties that it must have in order to make this reduction sufficient to derive our hardness results.

First, we simply notice that the rank of $\widehat{\lat} = \lat \oplus \lat^\dagger$ is the sum of the ranks of $\lat$ and $\lat^\dagger$. To prove the kind of fine-grained hardness results that we are after, we are only willing to increase the rank by a constant factor, so \emph{the rank of $\lat^\dagger$ must be at most $O(n)$}. (Of course, prior work did not have this restriction.)

Next, we notice that any $\widehat{\vec{y}} = (\vec{y}, \vec{y}^\dagger) \in \widehat{\lat}$ with $\|\vec{y} - \vec{t}\|_p \leq r$ and $\|\vec{y}^\dagger - \vec{t}^\dagger\|_p \leq r^\dagger$ satisfies $\|\widehat{\vec{y}} - \widehat{\vec{t}} \|_p \leq \widehat{r}$. We call these \emph{good vectors}, and we notice that there are at least $N_p(\lat^\dagger, r^\dagger, \vec{t}^\dagger)$ good vectors in the YES case.

Now, we worry about short vectors in $\widehat{\lat}$ in the NO case, i.e., non-zero $\widehat{\vec{y}} = (\vec{y}, \vec{y}^\dagger)$ with $\|\widehat{\vec{y}}\|_p \leq \widehat{r}$. Clearly, $\widehat{\vec{y}}$ will be short if $\|\vec{y}\|_p \leq r$ and $\|\vec{y}^\dagger\|_p \leq r^\dagger$. Therefore, the number of short vectors is at least
\[
 N_p(\lat, r, \vec{0}) \cdot N_p(\lat^\dagger, r^\dagger, \vec{0})  \geq 2^{\Omega(n)} \cdot N_p(\lat^\dagger, r^\dagger,  \vec{0}) \geq 2^{\Omega(n^\dagger)} \cdot N_p(\lat^\dagger, r^\dagger,  \vec{0})
 \; ,
\]
where we have used the fact that $n^\dagger = O(n)$ and the fact that the input $\CVP_p$ instances that interest us have $2^{\Omega(n)}$ short vectors. (This is not true in general, but it is true of most $\CVP_p$ instances resulting from hardness proofs.)
Since the number of good vectors in the YES case is potentially only $N_p(\lat^\dagger, r^\dagger, \vec{t}^\dagger)$, \emph{our gadget lattice must satisfy}
\begin{equation}
\label{eq:more_close_than_short_intro}
N_p(\lat^\dagger, r^\dagger, \vec{t}^\dagger) \geq 2^{\Omega(n^\dagger)} \cdot N_p(\lat^\dagger, r^\dagger, \vec0)
\; .
\end{equation}
Though this in itself is not sufficient to make our reduction work, it is the most important feature that a gadget lattice must have.  Indeed, we show in Corollary~\ref{cor:gadget_to_hardness} that a slightly stronger condition is sufficient to prove Gap-ETH hardness.
(This property and various variants are sometimes called \emph{local density}, and they play a key role in many hardness proofs for $\SVP_p$.)

However, short vectors are no longer our only concern. We also have to worry about close vectors that are not good vectors, i.e., vectors $\widehat{\vec{y}} = (\vec{y}, \vec{y}^\dagger)$ in the NO case such that $\|\widehat{\vec{y}} - \widehat{\vec{t}}\|_p \leq \widehat{r}$ but $\|\vec{y} - \vec{t}\|_p > r$. We call such vectors \emph{impostors}. Impostors certainly can exist in general, but our sparsification procedure will work on them just like any other vector. So, as long as our gadget lattice is chosen such that  the number of impostors in the NO case is significantly lower than the number of good vectors in the YES case, they will not trouble us. 

\subsection{Our techniques}
\label{sec:techniques}

We learned in the previous section that, in order to make our reduction work, it is necessary (though not always sufficient) that our gadget $(\lat^\dagger, \vec{t}^\dagger, r^\dagger)$ has exponentially more close vectors than short vectors. I.e., we need to find a family of gadgets that satisfies Eq.~(\ref{eq:more_close_than_short_intro}). Furthermore, we must somehow ensure that the the number of impostors in the NO case is exponentially lower than the number of good vectors in the YES case. 

\paragraph{The integer lattice, $\Theta_p$, and SETH-hardness.} To prove Theorem~\ref{thm:SETH_intro}, we take $\lat^\dagger := \Z^{n^\dagger}$, $\vec{t}^\dagger := (1/2,\ldots, 1/2) \in \R^{n^\dagger}$, and $r^\dagger := \dist_p(\vec{t}^\dagger, \Z^{n^\dagger}) = (n^\dagger)^{1/p}/2$. Notice that, by taking $r^\dagger = \dist_p(\vec{t}^\dagger, \lat^\dagger)$, we ensure that there simply are no impostors in the NO instance (i.e., when $\|\vec{y} - \vec{t} \|_p > r$, we can never have $\|(\vec{y}, \vec{y}^\dagger)-(\vec{t}, \vec{t}^\dagger)\|_p^p \leq r^p + (r^\dagger)^p$).\footnote{We note that any gadget that allows us to use $r^\dagger = \dist_p(\vec{t}^\dagger, \lat^\dagger)$ must satisfy quite rigid requirements. We need exponentially many vectors that are all \emph{exact} closest vectors, and we still must satisfy Eq.~\eqref{eq:more_close_than_short_intro}.}

To prove that our reduction works, we wish to show that the ratio
\[
\frac{N_p(\Z^{n^\dagger}, r^\dagger, \vec{t}^\dagger)}{N_p(\Z^{n^\dagger}, r^\dagger, \vec0)}
\]
is (exponentially) large.
Of course, the numerator is easy to calculate. It is $|\{ 0,1\}^{n^\dagger}| = 2^{n^\dagger}$. So, we wish to prove that
\begin{equation}
\label{eq:fewer_than_2n}
N_p(\Z^{n^\dagger}, r^\dagger, \vec0) \ll 2^{n^\dagger}
\; .
\end{equation}

Unfortunately, Eq.~\eqref{eq:fewer_than_2n} does not hold for all $\ell_p$ norms. For example, for $p = 2$, consider the points in $\{-1,0,1\}^{n^\dagger}$ with $n^\dagger/4$ non-zero coordinates, which have $\ell_2$ norm  $r^\dagger$. There are 
\[
2^{n^\dagger/4} \cdot \binom{n^\dagger}{n^\dagger/4} \approx 2^{n^\dagger/4} \cdot 4^{n^\dagger/4} \cdot (4/3)^{3n^\dagger/4} \approx 2.0867^{n^\dagger}
\]
such points. (In fact, this is a reasonable estimate for the exact value of $N_2(\Z^{n^\dagger}, r^\dagger ,\vec{0})$, which is $C^{n^\dagger + O(\sqrt{n^\dagger})}$ for $C = 2.0891\ldots$, as we show in \full{Section~\ref{sec:integer_points}}{the full version of this paper}.)
However, $N_p(\Z^{n^\dagger}, (n^\dagger)^{1/p}/2, \vec{t}^\dagger)$ is decreasing in $p$. So, one might hope that Eq.~\eqref{eq:fewer_than_2n} holds for slightly larger $p$.

To prove this, we wish to find a good upper bound on the number of integer points in a centered $\ell_p$ ball, $N_p(\Z^{n^\dagger}, r^\dagger, \vec{0})$.
A very nice way to do this uses the function
\[
\Theta_p(\tau) := \sum_{z \in \Z} \exp(- \tau |z|^p)
\; 
\]
for $\tau > 0$~\cite{MO90,EOR91}.\footnote{One can think of this as a variant of Jacobi's Theta function. In particular, with $p = 2$, this is Jacobi's Theta function with a slightly different parametrization. Computer scientists might be more familiar with the closely related function $\rho_s(\Z) := \Theta_2(\pi/s^2)$.}
Notice that
\[
\Theta_p(\tau)^{n^\dagger} = \sum_{z_1,  z_2,\ldots , z_n \in \Z} \exp(-\tau (|z_1|^p + \cdots + |z_{n^\dagger}|^p)) = \sum_{\vec{z} \in \Z^{n^\dagger}} \exp(-\tau \|\vec{z}\|_p^p)
\; .
\]
In particular,
\[
\Theta_p(\tau)^{n^\dagger} \geq \sum_{\stackrel{\vec{z} \in \Z^{n^\dagger}}{\|\vec{z}\|_p \leq r^\dagger}} \exp(-\tau \|\vec{z}\|_p^p) \geq \exp(-\tau (r^\dagger)^p) \cdot N_p(\Z^{n^\dagger}, r^\dagger, \vec{0})
\; .
\]
Rearranging and taking the infimum over $\tau$, we see that
\begin{equation}
	\label{eq:Theta_upper_bound}
	N_p(\Z^{n^\dagger}, r^\dagger, \vec{0}) \leq \inf_{\tau > 0} \exp(\tau (r^\dagger)^p) \Theta_p(\tau)^{n^\dagger}
	\; .
\end{equation}
We can relatively easily compute this value numerically and see that it is less than $2^{n^\dagger}$ for $p > p_0 \approx 2.1397$. (Indeed, we will see below that there is a nearly matching lower bound in a more general context. So, Eq.~\eqref{eq:Theta_upper_bound} is quite tight.)

To prove Theorem~\ref{thm:SETH_intro}, we can plug this very simple gadget into Khot's reduction described in Section~\ref{sec:Khot} to reduce the SETH-hard instances of $\CVP_p$ from~\cite{BGS17} to $\SVP_p$. 
To make the constant $C_p$ as tight as we can, we exploit the structure of these SETH-hard $\CVP_p$ instances. In particular, we observe that these instances themselves actually look quite a bit like our gadget, in that they are in some sense ``small perturbations'' of the integer lattice with the all one-halves point as the target. (\full{See Section~\ref{sec:gadget_Zn_all_halves}. }{}This is in fact quite common for the $\CVP_p$ instances resulting from hardness proofs.) This allows us to analyze the direct sum resulting from Khot's reduction very accurately in this case.

\paragraph{More $\Theta_p$ for $p  > 2$, and Gap-ETH hardness.} To extend our hardness results to all $p >2$, we need to construct a gadget with exponentially more close vectors than short vectors for such $p$. We again choose our gadget lattice as $\Z^{n^\dagger}$, but we now take $\vec{t}^\dagger = (t,t,\ldots, t) \in \R^{n^\dagger}$ for some $t \in (0,1/2]$, and we take $r^\dagger = C (n^\dagger)^{1/p}$ for some constant $C > 0$.

Our previous gadget was quite convenient in that it was very easy to count the number of close vectors, but for arbitrary $t$ and $r^\dagger$, it is no longer clear how to do this. Fortunately, $\Theta_p$ can be used for this purpose. In particular, we define
\[
\Theta_p(\tau; t) := \sum_{z \in \Z} \exp(-\tau |z-t|^p)
\; .
\]
By the same argument as before, we see that
\[
N_p(\Z^{n^\dagger}, r^\dagger, \vec{t}^\dagger) \leq \inf_{\tau > 0} \exp(\tau (r^\dagger)^p) \Theta_p(\tau; t)^{n^\dagger} = \big( \inf_{\tau > 0} \exp(\tau C^p) \Theta_p(\tau; t) \big)^{n^\dagger}
\; .
\]
But, we need a \emph{lower bound} on $N_p(\Z^{n^\dagger}, r^\dagger, \vec{t}^\dagger)$. To that end, we show that the above is actually quite tight. In particular,
\begin{equation}
\label{eq:Theta_approx_intro}
N_p(\Z^{n^\dagger}, r^\dagger, \vec{t}^\dagger) =  \big( \inf_{\tau > 0} \exp(\tau C^p) \Theta_p(\tau; t) \big)^{n^\dagger} \cdot 2^{-O(\sqrt{n^\dagger})}
\; .
\end{equation}
I.e., $\Theta_p$ tells us the number of integer points in an $\ell_p$ ball up to lower-order terms. (Eq.~\eqref{eq:Theta_approx_intro} was already proven for $p=2$ by Mazo and Odlyzko~\cite{MO90} and for all $p$ by Elkies, Odlyzko, and Rush~\cite{EOR91}.\full{ See Section~\ref{sec:integer_points} for the proof.}{})

It follows that there exists a $\vec{t}^\dagger$ and $r^\dagger$ with exponentially more close integer vectors than short integer vectors in the $\ell_p$ norm if and only if there exists a $\tau > 0$ and $t \in (0,1/2]$ such that $\Theta_p(\tau; t) > \Theta_p(\tau; 0)$. Furthermore, this holds if and only if $p > 2$.\full{ See Section~\ref{sec:integer_points} for the proof.}{}

So, to prove Theorem~\ref{thm:ETH_intro}, we start with the observation that approximating the Exact Set Cover problem is Gap-ETH-hard for some constant approximation factor $\eta < 1$. We then plug our gadget into Khot's reduction from constant-factor-approximate Exact Set Cover to $\SVP_p$. (This reduction uses $\CVP_p$ as an intermediate problem.) The above discussion explains why Eq.~\eqref{eq:more_close_than_short_intro} is satisfied. And, like Khot, we exploit the approximation factor $\eta$ to show that the number of impostors in a NO instance is much smaller than the number of good vectors in a YES instance.

\paragraph{Building gadgets in $\ell_2$ from lattices with high kissing number.} While we are not able to construct a gadget that satisfies Eq.~\eqref{eq:more_close_than_short_intro} in the $\ell_2$ norm, we show the existence of such a gadget under the reasonable assumption that for any $n^\dagger$, there exists a lattice $\cL^\dagger$ of rank $n^\dagger$ with exponentially many non-zero vectors of minimal $\ell_2$ norm. I.e., we show that such a gadget exists if there is a family of lattices with exponentially large kissing number. (We actually show that something potentially weaker suffices. See \full{Theorem~\ref{thm:kissing_gives_hardness}}{the full version of this paper}.)
 
To prove this, we show how to choose a $\vec{t}^\dagger$ and $r^\dagger < \lambda_1(\lat^\dagger)$ such that $N_p(\lat^\dagger, r^\dagger, \vec{t}^\dagger) \geq 2^{\Omega(n^\dagger)}$.
Indeed, we show that if we choose the vector $\vec{t}^\dagger$ uniformly at random from vectors of an appropriate length, then the expected number of lattice vectors within distance $r^\dagger$ from $\vec{t}^\dagger$ is exponential in $n^\dagger$. And, we again exploit the fact that  there is a constant-factor gap between the YES and the NO instances to show that the number of impostors in the NO instances is exponentially smaller than the number of good vectors in the YES instances.

\subsection{Direction for future work}
\label{sec:future}

Our dream result would be an explicit $2^{C n}$-time lower bound on approximate $\SVP_2$ for the approximation factors most relevant to cryptography (e.g., $\poly(n)$) for some not-too-small explicit constant $C > 0$, under a reasonable complexity-theoretic assumption. This seems very far out of reach. There are even complexity-theoretic barriers towards achieving this result, since $\SVP$ with these approximation factors cannot be NP-hard unless the polynomial-time hierarchy collapses~\cite{AharonovR04,Peikert08}. So, any proof of something this strong would presumably have to use a non-standard reduction (e.g., a non-deterministic reduction). Nevertheless, we can still dream of such a result and take more modest steps to at least get results closer to this dream.

One obvious such step would be to extend our hardness results to the $p = 2$ case, i.e., to show that there is no $2^{o(n)}$-time algorithm for $\SVP_2$ under reasonable purely complexity-theoretic assumptions (as opposed to our geometric assumption). We provide one potential route towards proving this in Theorem~\ref{thm:kissing_intro} (or its more general version in \full{Theorem~\ref{thm:kissing_gives_hardness}}{the full version}), but this would require resolving an older open problem in the geometry of numbers. Perhaps a different approach will prove to be more fruitful?

Alternatively, one could try to improve the approximation factor given by Theorem~\ref{thm:ETH_intro}. The currently known hardness of approximation proofs for $\SVP_p$ with large approximation factor (e.g., a large constant or superconstant) work by ``boosting'' the approximation factor via repeatedly taking the tensor product~\cite{Khot05svp,HRsvp}. I.e., given a family of lattices $\lat \subset \R^d$ for which we know that $\SVP_p$ is hard to approximate to within some small constant factor $\gamma > 1$, we argue that it is hard to approximate $\SVP_p$ to within a factor of $\gamma^k$ on the tensor product $\lat^{\otimes k}$ for some $k \geq 2$. Unfortunately, even a single tensor product increases the rank of the lattice quadratically. So, we cannot afford to use this technique to prove reasonable fine-grained hardness of approximation results. We therefore need a new technique.

Yet another direction would be to try to improve the constant $C_p$ in Theorem~\ref{thm:SETH_intro}. Perhaps the simple gadget that we use is not the best possible.

Finally, in a completely different direction, we note that Theorem~\ref{thm:SETH_intro} provides some additional incentive to study algorithms for $\SVP_p$ for $p \neq 2$ to improve the hidden (very large) constant in the $2^{O(n)}$ running time of existing algorithms. In particular, it would be interesting to see how close we can get to the lower bound given by Theorem~\ref{thm:SETH_intro}.

\subsection*{Acknowledgments}

The authors thank Huck Bennett, Vishwas Bhargav, Noam Elkies, Sasha Golovnev, Pasin Manurangsi, Priyanka Mukhopadhyay, and Oded Regev for helpful discussions. In particular, we thank Noam Elkies for pointing us to~\cite{EOR91} and Oded Regev for observing that the gadgets that we need are related to lattices with high kissing number.

	\section{Preliminaries}
	
	\label{sec:prelims}
We denote column vectors $\vec{x} \in \R^d$ by bold lower-case letters. Matrices $\basis \in \R^{d \times n}$ are denoted by bold upper-case letters, and we often think of a matrix as a list of column vectors. For $\vec{x} \in \R^{d_1}, \vec{y} \in \R^{d_2}$, we abuse notation a bit and write $(\vec{x}, \vec{y}) \in \R^{d_1+d_2}$ when we should technically write $(\vec{x}^T, \vec{y}^T)^T$. For $x \in \R$, we write
	\[
	\exp(x) := e^x = 1 + x + x^2/2 + x^3/6 + \cdots
	\; .
	\]
	Logarithms are base $e$.
	
	Throughout this paper, we consider computational problems over $\R^d$. Formally, we should specify a method of representing arbitrary real numbers, and our running times should depend in some way on the bit length of these representations and the cost of doing arithmetic in this representation. For convenience, we ignore these issues (in particular assuming that basic arithmetic operations always have unit cost), and we simply note that all of our reductions remain efficient when instantiated with any reasonable representation of $\R$.
	When we say that something is efficiently computable as a function of a dimension $d$, rank $n$, or cardinalities $m$, we mean that it is computable in time $\poly(d)$, $\poly(n)$, or $\poly(m)$, respectively (as opposed to polynomial in the logarithm of these numbers).

	\subsection{Lattice problems}
	
		\begin{definition}
			For any $1 \leq p \leq \infty$ and any $\gamma \geq 1$, \emph{the $\gamma$-approximate Shortest Vector Problem in the $\ell_p$ norm} ($\SVP_{p, \gamma}$) is the promise problem defined as follows. The input is a (basis for a) lattice $\lat \subset \R^d$ and a length $r > 0$. It is a YES instance if $\lambda_1^{(p)}(\lat) \leq r$ and a NO instance if $\lambda_1^{(p)}(\lat) > \gamma r$.
		\end{definition}
	
	\begin{definition}
		For any $1 \leq p \leq \infty$ and any $\gamma  \geq 1$, \emph{the $\gamma$-approximate Closest Vector Problem in the $\ell_p$ norm} ($\CVP_{p, \gamma}$) is the promise problem defined as follows. The input is a (basis for a) lattice $\lat \subset \R^d$, a target $\vec{t} \in\R^d$, and a distance $r > 0$. It is a YES instance if $\dist_p(\vec{t}, \lat) \leq r$ and a NO instance if $\dist_p(\vec{t}, \lat) > \gamma r$.
	\end{definition}
	
	When $\gamma  = 1$, we simply write $\SVP_p$ and $\CVP_p$. 
	We will need the following (simplified version of a) celebrated result, due to Figiel, Lindenstrauss, and Milman~\cite{FLM76}.
	
	\begin{theorem}[\cite{FLM76}]
		\label{thm:embedding}
		For any $\eps \in (0,1)$, $1 \leq p \leq 2$, and any positive integers $n$ and $m$ with $m \geq n/\eps^2$, there exists a linear map $f : \R^n \to \R^{m}$ such that for any $\vec{x} \in\R^n$,
		\[
			(1-\eps)\|\vec{x}\|_2 \leq \|f(\vec{x})\|_p  \leq (1+\eps) \|\vec{x}\|_2
			\; .
		\]
	\end{theorem}
	
	Regev and Rosen showed how theorems like this can be applied to obtain reductions between lattice problems in different norms~\cite{RR06}. Here, we only need the following immediate consequence of the above theorem. (The non-uniform reduction can be converted into an efficient randomized reduction and a similar result holds for $p > 2$, but we do not need this for our use case.)

	\begin{corollary}
		\label{cor:embedding_reduction}
		For any constants $\gamma_1 > \gamma_2 > 1$ and $1 \leq p \leq 2$, there is an efficient rank-preserving non-uniform reduction from $\SVP_{2,\gamma_1}$ in dimension $d$ to $\SVP_{p,\gamma_2}$ in dimension $O(d)$.
	\end{corollary}
	
	\subsection{Sparsification}
	\label{sec:sparsification}

A lattice vector $\vec{y} \in \lat$ is \emph{non-primitive} if $\vec{y} = k \vec{x}$ for some scalar $k > 1$ and lattice vector $\vec{x} \in \lat$. Otherwise, $\vec{y}$ is \emph{primitive}. (Notice that $\vec0$ is non-primitive.) For a radius $r > 0$, we write 
\[
\xi_p(\lat, r) := |\{ \vec{y} \in \lat \ : \ \text{$\vec{y}$ is primitive and } \|\vec{y}\|_p \leq r\}|/2
\]
for the number of primitive lattice vectors of length at most $r$ in the $\ell_p$ norm (counting $\pm \vec{y}$ only once). We will use the following generalization of a sparsification theorem from~\cite{DGStoSVP} to all $\ell_p$ norms.

\begin{theorem}[{\cite[Proposition 4.2]{DGStoSVP}}]
	\label{thm:sparsify}
	There is an efficient algorithm that takes as input (a basis for) a lattice $\lat \subset \R^d$ of rank $n$ and a prime $q \geq 101$ and outputs a sublattice $\lat' \subseteq \lat$ of rank $n$ such that for any radius $r < q \cdot \lambda_1^{(p)}(\lat)$ and any $1 \leq p \leq \infty$,
	\[
	\frac{N}{q} - \frac{N^2}{q^2} \leq \Pr[\lambda_1^{(p)}(\lat') \leq r] \leq \frac{N}{q}
	\; ,
	\]
	as long as $N \leq q/(20 \log q)$, 
	where $N := \xi_p(\lat, r)$ is the number of primitive lattice vectors of length $r$ in the $\ell_p$ norm (up to the sign). Furthermore, if $r \geq q \lambda_1^{(p)}(\lat)$, then $\lambda_1^{(p)}(\lat') \leq r$ always.
\end{theorem}

We note in passing that the algorithm works by taking a random linear equation $\inner{\vec{z}, \vec{a}} \equiv 0 \bmod q$ for uniformly random $\vec{z} \in \Z_q^n$ and setting $\lat'$ to be the set of lattice vectors whose coordinates in some arbitrary fixed basis satisfy this linear equation. (This idea was originally introduced by Khot.)

	\subsection{Fine-grained assumptions}
	\label{sec:fine-grained_prelims}
	
	Recall that, for integer $k \geq 2$, a $k$-SAT formula is the conjunction of clauses, where each clause is the disjunction of $k$ literals. I.e., $k$-SAT formulas have the form $\bigwedge_{i=1}^m \bigvee_{j=1}^k b_{i,j}$, where $b_{i,j} = x_k$ or $b_{i,j} = \neg x_k$ for some boolean variable $x_k$. 
	
	\begin{definition} 
		For any $k \geq 2$, the decision problem $k$-SAT is defined as follows. The input is a $k$-SAT formula. It is a YES instance if there exists an assignment to the variables that makes the formula evaluate to true and a NO instance otherwise.
	\end{definition}
	
	\begin{definition}
		For any $k \geq 2$, the decision problem Max-$k$-SAT is defined as follows. The input is a $k$-SAT formula and an integer $S \geq 1$. It is a YES instance if there exists an assignment to the variables such that at least $S$ of the clauses evaluate to true and a NO instance otherwise.
	\end{definition}
	
	Notice that $k$-SAT is a special case of Max-$k$-SAT.
	
	Impagliazzo and Paturi introduced the following celebrated and well-studied hypothesis concerning the fine-grained complexity of $k$-SAT~\cite{IP1999}.
	
	\begin{definition}[SETH]
		The (randomized) \emph{Strong Exponential Time Hypothesis} ((randomized) SETH) asserts that, for every constant $\eps > 0$, there exists a constant $k \geq 3$ such that there is no $2^{(1-\eps)n}$-time (randomized) algorithm for $k$-SAT formulas with $n$ variables.
	\end{definition}

	\begin{definition}
		For $\eta \in (0,1)$ and $k \geq 2$, the promise problem Gap-$k$-$\SAT_\eta$ is defines as follows. The input is a $k$-SAT formula with $m$ clauses. It is a YES instance if the formula is satisfiable, and it is a NO instance if the maximal number of simultaneously satisfiable clauses is strictly less than $\eta m$.
	\end{definition}
	
	Dinur~\cite{journals/eccc/Dinur16} and Manurangsi and Raghavendra~\cite{MR17} recently introduced the following natural assumption, called Gap-ETH. We also consider a non-uniform variant.
	
	\begin{definition}[Gap-ETH]
		The (randomized) Gap-Exponential Time Hypothesis ((randomized) Gap-ETH) asserts that there exists a constant $\eta \in (0,1)$ such that there is no (randomized) $2^{o(n)}$-time algorithm for Gap-$3$-$\SAT_\eta$ over $n$ variables.
		
		Non-uniform Gap-ETH asserts that there is no circuit family of size $2^{o(n)}$ for Gap-$3$-$\SAT_\eta$ over $n$ variables.
	\end{definition}
	
	\begin{definition}
		For $\eta \in (0,1)$, $k \geq 2$, and $C \geq 2$, the promise problem Gap-$k$-$\SAT_\eta^{\le C}$ is defined as follows. The input is a $k$-SAT formula such that each variable appears in at most $C$ clauses. It is a YES instance if the formula is satisfiable, and it is a NO instance if the maximal number of simultaneously satisfiable clauses is at most $\eta m$.
	\end{definition}
	
	We will need the following result due to Manurangsi and Raghavendra~\cite{MR17}.
	
	\begin{theorem}[{\cite{MR17}}]
		Unless Gap-ETH is false, there exist constants $\eta \in (0,1)$ and $C \geq 2$ such that there is no $2^{o(n)}$-time algorithm for $\text{Gap}$-$3$-$\SAT_\eta^{\le C}$.
	\end{theorem}
	
		\begin{definition}
			For $\eta \in (0,1)$, the promise problem $\ESC_\eta$ is defined as follows. The input consists of sets $S_1, \cdots, S_m \subseteq U$ with $|U| = k$ and a positive integer ``size bound'' $d \leq m$. It is a YES instance if there exist disjoint sets $S_{i_1}, \cdots, S_{i_\ell}$ such that $\bigcup_j S_{i_j} = U$ for some $\ell \le \eta d$.  It is a NO instance if for every collection of (not necessarily disjoint) sets $S_{i_1}, \cdots, S_{i_d}$, $\bigcup_j S_{i_j} \neq U$.
		\end{definition}	
		
		The following reduction is due to~\cite{PasinPrivate}.
		
		\begin{theorem}
			\label{thm:SAT_to_ESC}
			For any constant $C'>0$, and $\eta' \in (0,1)$, there is a polynomial-time Karp reduction from Gap-$3$-$\SAT_{\eta'}^{\le C'}$ on $n$ variables to $\ESC_\eta$ with $d := n/\eta$ and $m, k  \in [n, Cn]$ for some constants $C > 1$ and $\eta \in (0,1)$ depending only on $C'$ and $\eta'$. 
		\end{theorem}
		\begin{proof}
			The reduction takes as input a set of clauses $\cC_1, \cC_2, \cdots, \cC_t$, over a set of variables $x_1, \ldots, x_n$ where each variable is in at most $C'$ clauses. We assume without loss of generality that each variable or its negation is in at least one clause, and so $n/3 \le t \le C'n$.

			Define $U$ to be the set $\{\cC_1, \ldots, \cC_t, x_1, \ldots, x_n\}$. For each literal $b = x_i$ or $b = \neg x_i$ and for each set $S$ of clauses containing $b_i$, we create a set $S \cup \{x_i\}$ in our instance. I.e. a literal that is contained in exactly $r$ clauses will be contained in exactly $2^r$ sets.
			The reduction outputs YES if there exists an exact set cover of size at most $n$, and no, otherwise. 
			
			It is easy to see that the reduction is efficient and that $n \le k \le (C'+1)n$ and $n \le m \le 2^{C' + 1}n$. We now argue correctness.
			
			Suppose the Gap-$3$-$\SAT_{\eta'}^{\le C'}$ instance is a YES instance, i.e. the formula is satisfiable. Then there exists a satisfying assignment obtained by setting $b_1 = b_2 = \cdots = b_n = 1$, where each $b_i$ is either $x_i$ or $\neg x_i$. Thus, for all $i = 1, 2, \ldots, n$, let $S_i$ be the set of clauses containing $b_i$ but not containing any of $b_1, \ldots, b_{i-1}$. Clearly, each of these sets is disjoint, and $\cup_i S_i = \{\cC_1, \ldots, \cC_t\}$, since $b_1, \ldots, b_n$ is a satisfying assignment. Thus, the sets $S_i \cup \{x_i\}$ form an exact set cover of $U$ of size $n$.
			
			Suppose, on the other hand, that the Gap-$3$-$\SAT_{\eta'}^{\le C'}$ instance is a NO instance, i.e. any assignment satisfies at most $\eta' t$ clauses. Let $S_1, \ldots, S_{\ell}$ be a set cover of $U$, where the sets are not necessarily disjoint. We wish to show that $\ell \geq d = n/\eta$ for some constant $\eta \in (0,1)$. 
			
			Let $S(b)$ be the set of all clauses containing a literal $b$. Without loss of generality, we can assume that each set $S_i$ equals either $S(x_j) \cup \{x_j\}$ or $S(\neg x_j) \cup \{x_j\}$ for some $j$.
			The total number of variables for which $S(x_i) \cup \{x_i\}$ and $S(\neg x_i) \cup \{x_i\}$ are both in the set cover is at most $\ell - n$. Thus, the total number of clauses covered by $S_1, \ldots, S_\ell$ is at most $\eta' t + C' (\ell - n)$, so we must have $\eta' t + C'(\ell - n) \ge t$. This implies that
			\[
			\ell \ge \frac{t}{C'}(1 - \eta') + n \geq
			\left(1 + \frac{1 - \eta'}{3C'}\right)  \cdot n
			\;,\] 
			as needed.
		\end{proof}
	
	\section{A reduction from a variant of CVP to  SVP}
	
	As we discussed in Section~\ref{sec:Khot}, the ``naive reduction'' from $\CVP_{p, \gamma'}$ to $\SVP_{p, \gamma}$ simply takes a $\CVP$ instance consisting of a basis $\basis \in \R^{d \times n}$ for a lattice $\lat \subset \R^d$, target $\vec{t} \in \R^d$, and distance $r > 0$, and constructs the $\SVP$ instance given by the basis for $\lat'$ of the form
	\[
	\basis' := 
	\begin{pmatrix}
	\basis &-\vec{t}\\
	0 &s
	\end{pmatrix}
	\]
	and length $r' := (r^p + s^p)^{1/p}$, where $s > 0$. Notice that, if the input is a YES instance (i.e., $\dist_p(\vec{t}, \lat) \leq r$, then $\lambda_1^{(p)}(\lat') \leq r'$.
	
	If the input instance is a NO instance (i.e., if $\dist_p(\vec{t}, \lat) > \gamma' r$), then we call a non-zero vector $\vec{y}' = (\vec{y} - z \vec{t}, zs) \in \lat'$ \emph{annoying} if $\|\vec{y}'\|_p \leq \gamma r'$. As Khot showed, we can sparsify (as in Theorem~\ref{thm:sparsify}), to make this naive reduction work as long as there are significantly fewer annoying vectors than close vectors. We therefore define a rather unnatural quantity below that exactly counts the number of annoying vectors in a NO instance.
	
	For $1 \leq p < \infty$, and $\gamma \ge 1$, a lattice $\lat \subset \R^d$, target $\vec{t} \in \R^n$, and distances $r,s > 0$, we define
	\[
	A_{r, s, \gamma}^{(p)}(\vec{t}, \lat) := \sum_{z = 0}^{\gamma (r^p/s^p + 1)^{1/p}} N_p(\lat, (\gamma^p r^p - (z^p - \gamma^p) s^p)^{1/p}, z \vec{t}) - 1
		\; .
	\]
	Notice that $A_{r,s,\gamma}^{(p)}$ does in fact count the number of annoying vectors resulting from the above reduction (up to sign). In particular, the summand $N_p(\lat, (\gamma^p r^p - (z^p - \gamma^p) s^p)^{1/p}, z \vec{t})$ is the number of vectors $\vec{y}' = (\vec{y} - z\vec{t}, zs) \in \lat$ of length at most $\gamma r'$ for some fixed $z$.

We now define the class of $\CVP_p$ instances on which this sparsification-based reduction works.

\begin{definition}
	For $1 \leq p < \infty$, $A = A(n) \geq 0$ (the number of annoying vectors), $G = G(n) \geq 1$ (the number of ``good'' or close vectors), and $\gamma = \gamma(n) \ge 1$ (the approximation factor), the promise problem $(A, G)\text{-}\CVP_{p, \gamma}$ is defined as follows. The input is a (basis for a) lattice $\lat \subset \R^d$, target $\vec{t} \in \R^d$, and distances $r,s > 0$. It is a YES instance if $N_p(\lat, r, \vec{t}) \geq G$.  It is a NO instance if $A_{r, s, \gamma}^{(p)}(\vec{t}, \lat) \le \alpha$.
\end{definition}

Notice that the YES and NO instances of $(A, G)\text{-}\CVP_{p,\gamma}$ are disjoint when $A < G$, since $A_{r,s,\gamma}^{(p)}(\vec{t}, \lat) \geq N_p(\lat, r, \vec{t}) $.\footnote{We find it convenient to define the problem even for $A \geq G$ because it will not always be clear which of the two values is larger. Our results will always be vacuous when $A \geq G$. E.g., Theroem~\ref{thm:sparsification_reduction} requires $G \gg A$.} We drop the subscript $\gamma$ from $A_{r, s, \gamma}^{(p)}(\vec{t}, \lat)$, $(A, G)\text{-}\CVP_{p, \gamma}$ and $\SVP_{p,\gamma}$  if $\gamma = 1$. 

		Having defined $(A,G)$-$\CVP_{p, \gamma}$ specifically so that we can reduce it to $\SVP_{p,\gamma}$, we now present the reduction from $(A, G)\text{-}\CVP_{p,\gamma}$ to $\SVP_{p, \gamma}$. It essentially follows from the definition of $A_{r, s,\gamma}^{(p)}$ together with Theorem~\ref{thm:sparsify}.
	
	\begin{theorem}
		\label{thm:sparsification_reduction}
		For $1 \leq p < \infty$ and efficiently computable $A = A(n) \geq 1$, $G = G(n) \geq 1000A(n)$, and $\gamma = \gamma(n)\ge 1$, there is a (randomized) reduction from $(A, G)\text{-}\CVP_{p, \gamma}$ on a lattice with rank $n$ in $d$ dimensions to $\SVP_{p, \gamma}$ on a lattice with rank $n+1$ in $d+1$ dimensions that runs in time $\poly(d, \log A, \log G)$.
		\end{theorem}	
		\begin{proof}
			On input a basis $\basis$ for a lattice $\lat \subset \R^n$, a target vector $\vec{t} \in \R^n$, and distances $r,s > 0$, the reduction does the following.
			Let $\lat'$ be the lattice generated by
			\[
			\basis' := \begin{pmatrix}
				\basis & -\vec{t}\\
				0 & s
			\end{pmatrix}
			\; ,
			\]
			as above.
			Let $M := 10\sqrt{AG}$.
			The reduction does the following $\ell := \ceil{100 d \log M}$ times. It finds a prime $q$ with $10 M \log M \leq q \leq 20 M \log M$ and calls the procedure from Theorem~\ref{thm:sparsify}, receiving as output some new lattice $\lat''$. It then calls its \SVP oracle with input $\lat''$ and $r':=(r^p + s^p)^{1/p}$. Finally, it outputs YES if and only if the $\SVP$ oracle returned YES more than $\delta  \ell$ times, where 
			\[
			\delta := \frac{M}{20 q} - \frac{M^2}{200 q^2}  \geq \frac{1}{100\log M}
			\; .
			\]
			
			The running time is clear, as is the fact that the reduction increases both the dimension and rank by exactly one.
			
		If the input instance is a YES instance, then the number of vectors in $\lat'$ of the form $(\vec{v} - \vec{t}, s)$, where $\vec{v} \in \lat$, is $N_p(\lat, r, \vec{t}) \geq G$. These are primitive vectors in $\lat'$ and have length at most $r'$ (and there is no pair $\pm \vec{y}$ in this collection of vectors). I.e., there are at least $M/10$ primitive lattice vectors in $\lat'$ of length at most $r'$, and it follows from Theorem~\ref{thm:sparsify} that
		\[
			\Pr[\lambda_1^{(p)}(\lat'') \leq r'] \geq 2\delta
			\; .
		\]
		Then, by the Chernoff-Hoeffding bound, the oracle will output YES except with probability $\exp(-\Omega(d))$, as needed.
		
		If the input instance is a NO instance, then notice the number of primitive vectors in $\lat'$ of length at most $\gamma r'$ is at most $A_{r, s, \gamma}^{(p)}(\vec{t}, \lat) \le A$ (up to sign). Furthermore, the total number of vectors of length at most $\gamma r'$ (including non-primitive vectors) is at most $2A_{r,s,\gamma}^{(p)}(\vec{t}, \lat)+1 \leq 2A + 1$. In particular, this implies that $\lambda_1^{(p)}(\lat') > \gamma r'/(A+1) > \gamma r'/q$.\footnote{Suppose that $\lambda_1^{(p)}(\lat') \leq \gamma r'/(A+1)$, and let $\vec{v} \in \lat \setminus \{\vec0\}$ with $\|\vec{v}\|_p \leq \gamma r'/(A+1)$. Then, for every $z \in \{-A-1,-A,\ldots, A,A+1\}$, $z \vec{v}$ is a distinct lattice vector with $\|z \vec{v}\|_p \leq \gamma r'$, which contradicts the fact that there are at most $2A+1$ such vectors.}
		 So, we may apply Theorem~\ref{thm:sparsify}, and we have that
		\[
			\Pr[\lambda_1^{(p)}(\lat'') \leq \gamma r'] \leq \frac{A}{q} \leq \frac{\delta}{2}
			\; .
		\]
		The result again follows by the Chernoff-Hoeffding bound.
		\end{proof}		
\section{SETH-hardness of this variant of CVP (and therefore SVP)}
\label{sec:gadget_Zn_all_halves}

We now show that $(A, G)\text{-}\CVP_p$ is SETH-hard. We first observe that the SETH-hard $\CVP_p$ instances from~\cite{BGS17} ``have a copy of $\Z^n$ embedded in them.'' This fact will allow us to compute $A_{r,s}^{(p)}$ quite accurately.

\begin{theorem}[\cite{BGS17}]
	\label{thm:CVPSETH}
	For any constant $k \geq 2$, the following holds for all but finitely many values of $p \geq 1$. There is a Karp reduction from Max-$k$-SAT on $n$ variables to $\CVP_p$ on a rank $n$ lattice $\lat \subset \R^d$ such that the resulting
	$\CVP_p$ instance $(\basis, \vec{t}, r)$ has the form
	\[
	\basis =
	\begin{pmatrix}
	\Phi\\
	I_n
	\end{pmatrix}
	\; ,
	\]
	for some matrix $\Phi \in \R^{(d- n) \times n}$;
	\[
	\vec{t} = 
	\begin{pmatrix}
	t_1\\
	\vdots\\
	t_{d-n}\\
	1/2\\
	\vdots\\
	1/2
	\end{pmatrix}
		\; ,
	\]
	for some scalars $t_i \in \R$; and $r = (n+1)^{1/p}/2$. 
	Moreover, when $k = 2$, this holds for all $p \geq 1$, and for any $k \geq 2$, this holds for all odd integers $p \geq 1$.
\end{theorem}

We note the following easy corollary, which we can think of as either an application of Khot's gadget reduction (as described in Section~\ref{sec:Khot}) or simply as a  ``padded'' variant of Theorem~\ref{thm:CVPSETH}.

\begin{corollary}
	\label{cor:kSAT_to_beta_CVP}
	For any constant integer $k \geq 2$, the following holds for all but finitely many values of $p \geq 1$. For any efficiently computable integer $n^\dagger = n^\dagger(n) \leq \poly(n)$, there is a Karp reduction from Max-$k$-SAT on $n$ variables to $(A, G)\text{-}\CVP_p$ on a rank $n+n^\dagger(n)$ lattice
	with 
	\[
	A := \sqrt{n+n^\dagger} \cdot N_p(\Z^{n+n^\dagger},(r^p + 1)^{1/p}, \vec{0}) \qquad \text{ and } \qquad G := 2^{n^\dagger}
	\; ,
	\]
	where $r := (n+ n^\dagger + 1)^{1/p}/2$ and $\widehat{\vec{t}} := (1/2, \ldots, 1/2) \in \R^{n+n^\dagger}$. Moreover, when $k = 2$, this holds for all $p \geq 1$, and for any $k \geq 2$, this holds for all odd integers $p \geq 1$.
\end{corollary}
\begin{proof}
	It suffices to show how to convert the $\CVP_p$ instance from Theorem~\ref{thm:CVPSETH} into a $(A,G)\text{-}\CVP_p$ instance. To do this, we simply append the matrix $I_{n^\dagger}$ to the basis and $\vec{t}^\dagger := (1/2,\ldots, 1/2) \in \R^{n^\dagger}$ to the target. I.e., we construct,
	\[
		\basis :=
		\begin{pmatrix}
		\Phi & 0 \\
		I_n &0\\
		0 & I_{n^\dagger}
		\end{pmatrix}
		\; ,
		\]
		where $\Phi \in \R^{(d - n-n^\dagger) \times n}$ is as in Theorem~\ref{thm:CVPSETH}, and
		\[
		\vec{t} := 
		\begin{pmatrix}
		t_1\\
		\vdots\\
		t_{d-n-n^\dagger}\\
		1/2\\
		\vdots\\
		1/2
		\end{pmatrix} \in \R^{d}
		\; ,
		\]
		where $t_i \in \R$ are as in Theorem~\ref{thm:CVPSETH}. We simply take $s = 1$.
	
	Let $\lat := \lat(\basis) \subset \R^d$. Let $\lat' \subset \R^{d - n^\dagger}$ be the lattice generated by the basis without the additional identity matrix, and let $\vec{t}' \in \R^{d - n^\dagger}$ be the target without the additional coordinates. Notice that vectors in $\lat$ have the form $\vec{y} := (\vec{y}', \vec{z})$, where $\vec{y}' \in \lat'$ and $\vec{z} \in \Z^{n^\dagger}$. In particular,
	\[
	\|\vec{y} - \vec{t} \|_p^p = \|\vec{y}'-\vec{t}'\|_p^p + \|\vec{z} - \vec{t}^\dagger\|_p^p \geq \|\vec{y}'-\vec{t}'\|_p^p + n^\dagger/2^p
	\; .
	\]
	
	So, if the input Max-$k$-SAT instance is a YES instance, then $\dist_p(\vec{t}', \lat') \leq (n+1)^{1/p}/2$, and so  there are at least $2^{n^\dagger}$ distinct vectors in $\vec{y} \in \lat$ such that $\|\vec{y} - \vec{t}\|_p \leq r$. (In particular, all vectors of the form $(\vec{y}', \vec{z})$ with $\vec{z} \in \{0,1\}^{n^\dagger}$ and $\|\vec{y}' - \vec{t}'\|_p \leq (n+1)^{1/p}/2$ have this property.) Thus, the resulting $(A,G)\text{-}\CVP_p$ instance is a YES instance.

	On the other hand, if the input Max-$k$-SAT instance is a NO instance, then we have that $\dist_p(\vec{t}', \lat') > (n+1)^{1/p}/2$. I.e., $N_p(\lat, r, \vec{t}) = 0$. Therefore,
	\begin{align*}
	A_{r, 1}^{(p)}(\vec{t}, \lat) 
	&= \sum_{z = 0}^{(r^p + 1)^{1/p}} N_p(\lat, (r^p -z^p + 1)^{1/p}, z \vec{t}) -1 \\
	&\leq N_p(\Z^{n+n^\dagger},(r^p + 1)^{1/p}, \vec{0}) -1 + N_p(\lat, r, \vec{t})  + \sum_{z = 2}^{(r^p + 1)^{1/p}} N_p(\Z^{n+n^\dagger},(r^p - z^p + 1)^{1/p}, z \widehat{\vec{t}})\\
	&\leq r \cdot N_p(\Z^{n+n^\dagger},(r^p + 1)^{1/p}, \vec{0})	\; ,
	\end{align*}
	where we have used the fact that $N_p(\Z^{n+n^\dagger},(r^p - z^p + 1)^{1/p}, z \widehat{\vec{t}}) = 0$ for odd $z \geq 3$ and 
	\[
	N_p(\Z^{n+n^\dagger},(r^p - z^p + 1)^{1/p}, z \widehat{\vec{t}}) = N_p(\Z^{n+n^\dagger},(r^p - z^p + 1)^{1/p}, \vec0) \leq N_p(\Z^{n+n^\dagger},(r^p + 1)^{1/p}, \vec{0})
	\]
	 for even $z$.
	Thus, the resulting $(A,G)\text{-}\CVP_p$ instance is a NO instance. 
\end{proof}

In the next section, we show that $A \ll G$ if and only if $p > p_0 \approx 2.13972$.

		\subsection{Finishing the proof}
		\label{sec:integer_points_centered}
		
		It remains to bound the number of integer points in an $\ell_p$ ball centered at the origin.
		As in Section~\ref{sec:techniques}, for $1 \leq p < \infty$ and $\tau > 0$, we define
		\[
		\Theta_p(\tau) := \sum_{z \in \Z} \exp(-\tau |z|^p)
		\; .
		\]
		Notice that we can write $\Theta_p(\tau)^n$ as a summation over $\Z^n$,
		\[
		\Theta_p(\tau)^n = \sum_{\vec{z} \in \Z^n} \exp(-\tau \|\vec{z}\|_p^p)
		\; .
		\]
		In particular, for any radius $r > 0$ and $\tau > 0$, we have
		\[
		\Theta_p(\tau)^n \geq \sum_{\stackrel{\vec{z} \in \Z^n}{\|\vec{z}\|_p \leq r}} \exp(-\tau \|\vec{z}\|_p^p) \geq \exp(-\tau r^p) N_p(\Z^n, r, \vec0)
		\; .
		\]
		Rearranging and taking the minimum over all $\tau > 0$, we see that
		\begin{equation}
		\label{eq:tau_bound_on_integer_points_centered}
		N_p(\Z^n, r, \vec0) \leq \min_{\tau >0 } \exp(\tau r^p) \Theta_p(\tau)^n
		\; .
		\end{equation}
		(It is easy to see that the minimum is in fact achieved.)
		In Section~\ref{sec:integer_points}, we will show that this upper bound is actually quite tight (even in the more general settings of shifted balls here). Here, we use this bound to prove the following theorem.
		
		\begin{theorem}
			\label{thm:SETH_hardness_centered_theta}
			For any constant integer $k \geq 2$, the following holds for all but finitely many constants $p > p_0$. There is an efficient randomized reduction from Max-$k$-SAT on $n$ variables to $\SVP_p$ on a lattice of rank $\ceil{C_p n + \log^2 n}$, where 
			\[
			C_p := \frac{1}{1-\log_2W_p} \qquad \text{ and } \qquad W_p := \min_{\tau > 0} \exp(\tau/2^p) \Theta_p(\tau)
			\; .
			\] 
			Here, $p_0 \approx 2.13972$ is the unique solution to the equation $W_{p_0} = 2$. Moreover, when $k = 2$, this holds for all $p > p_0$, and for any $k \geq 2$, this holds for all odd integers $p \geq 3$.
			
			In particular, for every $\eps > 0$, for all but finitely many $p > p_0$ (including all odd integers $p \geq 3$) there is no $2^{n/(C_n + \eps)}$-time algorithm for $\CVP_p$ unless SETH is false.
		\end{theorem}
		\begin{proof}
			Let $n^\dagger := \ceil{C_p n + \log^2 n} - n - 1$
Then, by Corollary~\ref{cor:kSAT_to_beta_CVP} together with Theorem~\ref{thm:sparsification_reduction}, it suffices to show that 
			\[
					N_p(\Z^{n+n^\dagger}, \widehat{r}, \vec0) \leq 2^{n^\dagger-10}/\sqrt{n+n^\dagger}
			\]
			for sufficiently large $n$,
			where $\widehat{r} := (n+n^\dagger + 2^{p+1})^{1/p}/2$.
			
			Let $\tau_p > 0$ such that $W_p := \exp(\tau_p/2^p) \Theta_p(\tau_p)$. (One can check that $\tau_p$ exists, e.g., by differentiating $\exp(\tau_p/2^p) \Theta_p(\tau_p)$ with respect to $\tau$.) By Eq.~\eqref{eq:tau_bound_on_integer_points_centered}, we have
			\[
			N_p(\Z^{n+n^\dagger}, \widehat{r}, \vec0) \leq \exp(\tau_p \widehat{r}^p) \Theta_p(\tau_p)^{n + n^\dagger} \leq \exp(2\tau_p) \cdot W_p^{n + n^\dagger} \leq \exp(2\tau_p + 1) \cdot 2^{n^\dagger - \log^2 n} \cdot W_p^{\log^2 n}
			\; .
			\]
			The result follows by noting that $W_p < 2$ so that for sufficiently large $n$,
			$(2/W_p)^{\log^2 n} \geq \exp(2\tau_p + 20) \sqrt{n+n^\dagger}$.
		\end{proof}

	Finally, we compute a simple bound on $C_p$. In particular, this implies the claim that $C_p \to 1$ as $p \to \infty$.
		
		\begin{claim}
			For any $p \geq 3$, we have
			\[
			C_p <  \frac{1}{1-2^{-p}(p +\log_2(3 e))} 
			\; .
			\]
		\end{claim}
		\begin{proof}
			We have 
			\[
			W_p := \min_{\tau > 0} \exp(\tau/2^p) \Theta_p(\tau) = \min_{x > 1} x^{2^{-p}} \cdot (1 + 2x^{-1} + 2x^{-2^p} + 2x^{-4^p} + \cdots)
			\; .
			\]
			Fix $x := 3 \cdot 2^{p}$. Then, we have
			\[
			W_p \leq x^{2^{-p}} \cdot (1 + 2x^{-1} + 2x^{-2^p} + 2x^{-4^p} + \cdots ) < x^{2^{-p}} \cdot (1 + 3x^{-1}) = x^{2^{-p}} \cdot (1+ 2^{-p})
			\; .
			\]
			Therefore,
			\[
			\log_2 W_p < 2^{-p}\log_2 x + \log_2 (1+2^{-p}) < 2^{-p}(p +\log_2(3 e))
			\; ,
			\]
			so that
			\[
			C_p = \frac{1}{1-\log_2 W_p} < \frac{1}{1-2^{-p}(p +\log_2(3 e))} 
			\; ,
			\]
			as needed.
		\end{proof}		
	
	\section{Gap-ETH-hardness via a gadget}
	\label{sec:Gap-ETH}

The following theorem shows how to use a certain gadget lattice $\lat^\dagger \subset \R^{d^\dagger}$ and target $\vec{t}^\dagger \in \R^{d^\dagger}$ with certain properties to reduce $\ESC_\eta$ to $(A, G)\text{-}\CVP_p$. In particular, the ratio of the number of ``close points'' in $\lat^\dagger$ to $\vec{t}^\dagger$ compared to the number of ``short points'' in $\lat^\dagger$ should be larger than the total number of short points in $\Z^m$. (For $p > 2$, we construct such a gadget in Section~\ref{sec:integer_points} that will be sufficient to prove the Gap-ETH-hardness of $\SVP_p$. For $1 \leq p \leq 2$, we do not know of such a gadget, but in Section~\ref{sec:gapeth_l2}, we will show that one exists under plausible conjectures.)

For $1 \leq p \leq \infty$, a lattice $\lat \subset \R^n$, and radius $r > 0$, we define the maximal density at radius $r$ of $\lat$ as 
\[
D_p(\lat, r) := \max_{\vec{t} \in \R^n} N_p(\lat, r, \vec{t})
\; .
\]
We observe the trivial fact that $D_p(\lat, r)$ is non-decreasing in $r$.

\begin{theorem}
	\label{thm:ESC_to_CVP}
	For any $p \geq 1$, constant $\eta \in (0,1)$ and $\gamma \ge 1$, there is a Karp reduction from $\ESC_\eta$ on $m$ sets with size bound $d$ to $(A, G)\text{-}\CVP_{p,\gamma}$ on a lattice of rank $m + n^\dagger$ that requires as auxiliary input a gadget consisting of a lattice $\lat^\dagger \subset \R^{d^\dagger}$ of rank $n^\dagger$, target $\vec{t}^\dagger \in \R^{d^\dagger}$, and distances $r \geq d^{1/p}$ and $s > 0$ for any
	\[
	A \geq N_p(\Z^m, r^*, \vec0) \cdot \big(N_p(\cL^\dagger, r^*, \vec0) 
	+ (r^*/s) \cdot D_{p}(\lat^\dagger, ((r^*)^p - d)^{1/p}) \big)
	\; ,
	\]
	and $G \leq N_p(\cL^\dagger, (r^p - \eta d)^{1/p}, \vec{t}^\dagger)$, where $r^* := \gamma(r^p + s^p)^{1/p}$.
\end{theorem}
\begin{proof}
	The reduction takes as input $S_1, \cdots, S_m \subseteq U= \{u_1, \ldots, u_k\}$ with $\bigcup S_i = U $, size bound $d \leq m$, a lattice $\lat^\dagger \subset \R^{d^\dagger}$ with basis $\basis^\dagger$, target $\vec{t}^\dagger \in \R^{d^\dagger}$, and distances $r, s > 0$ and behaves as follows.
	We first define the intermediate $\CVP$ instance consisting of the lattice $\widehat{\lat} := \lat(\widehat{\basis}) \subset \R^{m + k}$ and target $\widehat{\vec{t}} \in \R^{m + k}$ given by
	\[ \widehat{\basis} := (\widehat{\vec{b}}_1, \ldots, \widehat{\vec{b}}_m) = 
	\begin{array}{c}
	\begin{array}{cccccccccc} 
	S_1 \cdots S_j&  \cdots & \cdots & S_m & & & & & 
	\end{array} \\
	\begin{array}{c|cccc|cc|c|c}
	\cline{2-5}\cline{8-8}
	u_1& & \vdots & & & & & r^* &  \\
	\vdots  & & \vdots & &  & & &  r^*&  \\
	u_i& \cdots & r^* &\text{ if }u_i \in S_j & & & & \vdots&  \\
	\vdots &  & 0& \text{ otherwise}& & & & \vdots& \\
	\vdots & & & & & & & \vdots& =\widehat{\vec{t}}\\
	u_k & & & & & & &  r^*& \\
	\cline{2-5} \cline{8-8}
	& 1& & & & & & & \\
	& & \ddots& & & & & & \\
	& & & \ddots& & & & & \\
	& & & & 1& & & & \\
	\cline{2-5}\cline{8-8}
	\end{array}\\   
	\end{array}
	\;.\]
	The reduction then constructs the $(A,G)\text{-}\CVP$ instance consisting of a lattice $\lat := \lat(\basis) \subset \R^{m + k + d^\dagger}$, $\vec{t} \in \R^{m + k + d^\dagger}$, and the distances $r,s$, where
	\[
	\basis :=
	\begin{pmatrix}
	\widehat{\basis} & 0 \\
	0 & \basis^\dagger
	\end{pmatrix}
	\; ,
	\]
	and 
	\[
	\vec{t} :=
	\begin{pmatrix}
	\widehat{\vec{t}}\\
	\vec{t}^\dagger
	\end{pmatrix}
	\; .
	\]
	The reduction then outputs YES if the $(A, G)\text{-}\CVP_p$ oracle on input $(\basis, \vec{t}, r, s)$ outputs YES, and NO, otherwise. 
	
	Suppose the input is a YES instance, and let $i_1, \ldots, i_j$ with $j \leq \eta d$ such that the $S_{i_\ell}$ are disjoint with $\bigcup S_{i_\ell} = U$. Let $\widehat{\vec{v}} :=  \widehat{\vec{b}}_{i_1} + \cdots + \widehat{\vec{b}}_{i_j}$ be the corresponding vector in $\widehat{\lat}$. Notice that $\|\widehat{\vec{v}} - \widehat{\vec{t}}\|_p^p = j \leq \eta d$. Therefore, for any $\vec{v}^\dagger \in \lat^\dagger$ with $\|\vec{v}^\dagger\|_p^p \leq r^p - \eta d$, the vector $\vec{v} := (\widehat{\vec{v}}, \vec{v}^\dagger)$ is in $\lat$ with $\|\vec{v}\|_p \leq r$. So,
	\[
	N_p(\lat, r, \vec{t}) \geq N_p(\lat^\dagger, (r^p - \eta d)^{1/p}, \vec{t}^\dagger) \geq G
	\;,
	\]
	as needed.
	
	Now, suppose the input is a NO instance. Then, we wish to show that
	\[
	\sum_{\ell = 0}^{\gamma (r^p/s^p + 1)^{1/p}} N_p(\cL, (\gamma^p r^p - (\ell^p - \gamma^p) s^p)^{1/p}, \ell \cdot \vec{t})  \leq A
	\; .
	\]
	We first bound the $\ell = 0$ term as
	\begin{align*}
	N_p(\cL, r^*) &\leq N_p( \Z^m \oplus \cL^\dagger, r^*, \vec0)  \leq  N_p(\Z^m, r^*, \vec0) \cdot N_p(\lat^\dagger, r^*, \vec0)
	\;.
	\end{align*}
	
	Turning to the $\ell \geq 1$ terms, let $\vec{v} = (\widehat{\vec{v}} - \ell \cdot \widehat{\vec{t}}, \vec{v}^\dagger - \ell \cdot \vec{t}^\dagger) \in \lat$ with $\ell \geq 1$, and suppose that $\|\vec{v} - \ell \cdot \vec{t}\|_p^p \leq \gamma^p r^p - (\ell^p - \gamma^p)s^p$.
	Let $\widehat{\vec{v}} = \sum_{i=1}^m z_i \widehat{\vec{b}}_i$. First, notice that $\|(z_1,\ldots, z_m)\|_p^p \leq \gamma^p r^p - (\ell^p - \gamma^p)s^p \leq (r^*)^p$ because of the ``identity matrix gadget'' at the bottom of $\widehat{\basis}$. Furthermore, if there are at most $d$ non-zero $z_i$'s, $z_{i_1}, \ldots, z_{i_{j}}$, then since the input is a NO instance, there must be an element $u \in U$ not contained in $S_{i_1} \cup \cdots \cup S_{i_j}$. Thus, the $u$th coordinate of $\widehat{\vec{v}} - \ell \cdot \widehat{\vec{t}}$ is at least $r^*$, and we cannot possibly have $\|\vec{v}\|_p^p \leq \gamma^p r^p - (\ell^p - \gamma^p) s^p$. 
	
	So, it must be the case that there are at least $d$ non-zero $z_i$. In particular, $\|(z_1,\ldots, z_m)\|_p^p \geq d$, so we must have $\|\vec{v}^\dagger - \ell \cdot \vec{t}^\dagger\|_p^p \leq \gamma^p r^p - d - (\ell^p - \gamma^p)s^p < (r^*)^p -d$. Therefore,
	\begin{align*}
	N_p(\cL, (\gamma^p r^p - (\ell^p - \gamma^p) s^p)^{1/p}, \ell \cdot \vec{t}) 
	&\leq N_p(\Z^m, r^*, \vec0) \cdot N_p(\lat^\dagger, ((r^*)^p -d)^{1/p}, \ell \cdot \widehat{\vec{t}})\\
	&\leq N_p(\Z^m, r^*, \vec0) \cdot D_p(\lat^\dagger, ((r^*)^p -d)^{1/p})
	\; .
	\end{align*}
	The result follows by noting that the total contribution of the terms with $\ell \geq 1$ is at most $r^*/s$ times this quantity (since there are at most $r^*/s$ such terms).
\end{proof}

Our goal is now to construct a useful gadget $\lat^\dagger, \vec{t}^\dagger, r^\dagger$ for Theorem~\ref{thm:ESC_to_CVP}. In particular we wish to find a gadget with $\rank(\lat^\dagger) = O(n)$ and $G \gg A$. In the following rather technical lemma, we show that such a gadget exists if there exists any lattice with ``an exponential gap between the number of close vectors and the number of short vectors.''

\begin{lemma}
\label{lem:exp_gap}
Suppose that for some constants $p \ge 1$, $\eps \in (0,1/2)$, and $\beta > 1$, the following holds. For every sufficiently large integer $n$, there exists a lattice $\lat_n \subset \R^{d_n}$ with $\rank(\lat_n) = n$ and $d_n \le \poly(n)$, target $\vec{t}_n \in \R^{d_n}$, and radius $r_n > 0$ such that
		\begin{equation}
		\label{eq:more_close_than_short}
		N_p(\lat_n, (1-\eps)^{1/p} \cdot r_n, \vec{t}_n) \geq \beta^n \cdot N_p(\lat_n, r_n, \vec0)
		\; .
		\end{equation}
		Then, for any constants $C \geq 1$ and $\eta \in (2\eps^2,1)$, there exist constants $\gamma > 1$, $C^\dagger > 0$ such that the following holds. 
		\begin{enumerate}
			\item \label{item:non-uniform} For any sufficiently large integers $m$ and $d$ satisfying $m/C \leq d \leq C d$, there exist distances $r, s > 0$ and $\vec{t}^\dagger \in \R^{d_{n^\dagger}}$ for $n^\dagger := \ceil{C^\dagger m}$
			such that 
			\begin{equation}
			\label{eq:good_gadget}
			N_p(\Z^m, r^*, \vec0) \cdot \big(N_p(\cL^\dagger, r^*, \vec0) 
			+ (r^*/s) \cdot D_{p}(\lat^\dagger, ((r^*)^p - d)^{1/p}) \big) < 2^{-m} \cdot N_p(\cL^\dagger, (r^p - \eta d)^{1/p}, \vec{t}^\dagger)
			\; ,
			\end{equation}
			where $r^* := \gamma (r^p + s^p)^{1/p}$ and $\lat^\dagger := \alpha \lat_{n^\dagger}$ for some $\alpha > 0$.
			\item \label{item:uniform}  If we also have 
			\begin{equation}
			\label{eq:target_not_dumb}
			N_p(\lat_n, (1-\eps)^{1/p}r_n, \vec{t}_n) \geq  \beta^n \cdot D_p(\lat_n, (1-\eps/\sqrt{\eta})^{1/p}r_n)
			\; ,
			\end{equation}
			 then we can take $r = (1-\eps/2)^{1/p}\alpha r_{n^\dagger}$, $\vec{t}^\dagger = \alpha \vec{t}_{n^\dagger}$, $s = 1$, and $\alpha = (2\eta d/(\eps r_{n^\dagger}^p))^{1/p}$.
		\end{enumerate}
\end{lemma}
\begin{proof} We prove Item~\ref{item:uniform} first. We take $\lat^\dagger = \alpha \lat_{n^\dagger}$, $r = (1-\eps/2)^{1/p}\alpha r_{n^\dagger}$, $\vec{t}^\dagger = \alpha \vec{t}_{n^\dagger}$, $s = 1$, and $\alpha = (2\eta d/(\eps r_{n^\dagger}^p))^{1/p}$, as above. Notice that $r^p = 2(1-\eps/2)\eta d/\eps$. 
	We choose 
	\[
	\gamma^p = 1+ \min \{ 1/100 , \   (1/\sqrt{\eta}-1)^2/2\}\cdot \eps > 1
	\; .
	\]
	We assume without loss of generality that $\eta d \geq 10$, and $\eps d(1-\sqrt{\eta})^2 \geq 2\gamma^p$
	
	Notice that $r^* = O(m^{1/p})$.
	Thus, there is some constant $\tC$ such that 
	\[
	N_p(\Z^m, r^*, \vec0) \leq \tC^m
	\; .
	\]
	(This follows, e.g., from Eq.~\eqref{eq:Theta_upper_bound}.) Furthermore, notice that 
	\[
	(r^*)^p = \gamma^p (r^p + 1) \leq ((1-\eps/2) \alpha^p r_{n^\dagger}^p + 1) \cdot (1+\eps/100) \leq \alpha^p r_{n^\dagger}^p
	\; ,
	\]
	where the last inequality uses the assumption that $\eps \alpha^p r_{n^\dagger}^p/2 = \eta d \geq 10$.
	Therefore, by Eq.~\eqref{eq:more_close_than_short},
	\[
	N_p(\lat^\dagger, r^*, \vec0) \leq \beta^{-C^\dagger m} N_p(\lat_{n^\dagger}, (1-\eps)^{1/p} \cdot r_{n^\dagger}, \vec{t}_{n^\dagger}) = \beta^{-C^\dagger m} N_p(\cL^\dagger, (r^p - \eta d)^{1/p}, \vec{t}^\dagger)
	\; ,
	\]
	where the last equality uses the fact that 
	$
	\alpha^p (1-\eps)r_{n^\dagger}^p = r^p \cdot (1-\eps)/(1-\eps/2) = r^p - \eta d
	$.
	Finally, we note that 
	\begin{align*}
	(r^*)^p - d
		&\leq 2\eta d \cdot \frac{\cdot (1-\eps/2) \cdot (1+ (1/\sqrt{\eta}-1)^2\cdot \eps/2)}{\eps} - d + \gamma^p \\
		&= 2\eta d \cdot \frac{(1-\eps/\sqrt{\eta}) }{\eps} - \frac{\eps d}{2} \cdot (1-\sqrt{\eta})^2 + \gamma^p\\
		&\leq 2\eta d \cdot \frac{(1-\eps/\sqrt{\eta}) }{\eps} \\
		&= (1-\eps/\sqrt{\eta}) \cdot \alpha^p r_{n^\dagger}^p
	\; .
	\end{align*}
	Therefore, applying Eq.~\eqref{eq:target_not_dumb}, we have
	\[
	D_{p}(\lat^\dagger, ((r^*)^p - d)^{1/p}) \leq D_p(\lat_{n^\dagger}, (1-\eps/\sqrt{\eta})^{1/p}r_{n^\dagger}) \leq \beta^{-C^\dagger m} N_p(\lat_{n^\dagger}, (1-\eps)r_{n^\dagger}, \vec{t}_{n^\dagger})
	\; .
	\]
	Putting everything together, we see that it suffices to take $C^\dagger > 0$ to be a large enough constant so that $\tC^m\beta^{-C^\dagger m} < 2^{-m}/(1+r^*/s)$.
	
	We now move to proving Item~\ref{item:non-uniform}. By Item~\ref{item:uniform}, it suffices to find some new family of targets $\vec{t}_1',\vec{t}_2',\ldots,$ and radii $r_1',r_2',\ldots,$ satisfying Eqs.~\eqref{eq:more_close_than_short} and~\eqref{eq:target_not_dumb}, perhaps for some new constants $\eps' \in (0,1)$ and $\beta' > 1$. We would of course like to simply take $\vec{t}_n' = \vec{t}_n$ and $r_n' = r_n$, but we have to worry about the possibility that $D_{p}(\lat_n, (1-\eps/\sqrt{\eta}) r_n)$ is not much smaller than $N_p(\lat_n, (1-\eps)^{1/p}r_n, \vec{t}_n)$. Intuitively, this can only happen if either (1) $D_{p}(\lat_n, (1-\eps)^{1/p}r_n) \gg N_p(\lat_n, (1-\eps)^{1/p}r_n, \vec{t}_n)$, in which case we should clearly replace $\vec{t}_n$ with $\vec{t}_n'$ such that $N_p(\lat_n, (1-\eps)^{1/p}r_n, \vec{t}_n') = D_p(\lat_n, (1-\eps)^{1/p}r_n)$; or (2) $D_p(\lat_n, (1-\eps)^{1/p}r_n') \approx D_p(\lat_n, (1-\eps)^{1/p}r_n)$, for some $r_n' < r_n$, in which case we should clearly replace $r_n$ with $r_n'$. So, intuitively, as long as $\vec{t}_n$ and $r_n$ are ``reasonable,'' we should be done.
	
	To make this rigorous, let $N_n := N_p(\lat_n, r_n, \vec0)$, $\eps_0 := \eps$, and $r_n^{(-1)} = r_n$. For $i = 1,\ldots, \ell+1$, let $\eps_i := \eps_{i-1}/\sqrt{\eta}$ and $r_n^{(i)} := (1-\eps_{i-1}) \cdot r_n^{(i-1)} $. Here, $\ell$ is the largest integer such that $r_n^{(\ell)} > r_n/2$. In particular, $\ell$ is a constant. Let $N_n^{(i)} := D_p(\cL_n, r_n^{(i)})$. It follows that $N_{n}^{(\ell+1)} \le N_n$, since if there were $N_n+1$ distinct lattice vectors $\vec{v}_1, \ldots, \vec{v}_{N_n+1}$ at distance $r_n/2$ from any vector $\vec{t}$ then, by triangle inequality, there would necessary be $N_n+1$ distinct lattice vectors $\vec{v}_i - \vec{v}_1$ for $i = 1,\ldots, N_n +1$ of length at most $r_n$, contradicting the definition of $N_n$. 
	
	 By Eq.~\eqref{eq:more_close_than_short}, we have
	$N_n^{(0)} \ge \beta^{n^\dagger} \cdot N_n$. 
	Thus, there exists an $i \in \{0, \ldots, \ell\}$ such that 
	\[
	\frac{N_n^{(i)}}{N_n^{(i+1)}} \ge (\beta')^{n} \;,
	\]
	and 
	\[
	\frac{N_n^{(i)}}{N_n} \ge (\beta')^{n} \;,
	\]
	where $\beta' := \beta^{1/(\ell+1)} > 1$. Fix such an $i$. Then, taking $\eps' := \eps_{i-1}$, $r_n' := r_n^{(i-1)}$, and $\vec{t}_n'$ to be any vector satisfying $N_p(\cL, r_n^{(i)}, \vec{t}_n') = D_p(\cL, r_n^{(i)})$ gives the result.
\end{proof}

Putting everything together, we get the following conditional result.

\begin{corollary}
	\label{cor:gadget_to_hardness}
	Suppose that for some constants $p \ge 1$, $\eps \in (0,1/2)$, and $\beta > 1$, the following holds. For every sufficiently large integer $n$, there exists a lattice $\lat_n \subset \R^{d_n}$ with $\rank(\lat_n) = n$ and $d_n \le \poly(n)$, target $\vec{t}_n \in \R^{d_n}$, and radius $r_n > 0$ such that
			\begin{equation}
			\label{eq:more_close_than_short_corollary}
			N_p(\lat_n, (1-\eps)^{1/p} \cdot r_n, \vec{t}_n) \geq \beta^n \cdot N_p(\lat_n, r_n, \vec{0})
			\; .
			\end{equation}
	Then for any constant $C \geq 1$ and $\eta \in (0,1)$, there is a constant $\gamma > 1$, such that there is an efficient (non-uniform) reduction from from Gap-$3$-$\SAT_{\eta}^{\le C}$ on $n$ variables to $\SVP_{p, \gamma}$ on a lattice of rank $O(n)$ and dimension $O(n)+d_{O(n)}$, for some constant $\gamma > 1$.
	
	Furthermore, if $\lat_n$, $\vec{t}_n$, and $r_n$ are computable in time $\poly(n)$ and	for any constant $\delta \in (\eps, 1)$, there exists a $\beta' > 1$ such that
	\begin{equation}
	\label{eq:target_not_dumb_corollary}
	N_p(\lat_n, (1-\eps)^{1/p}r_n, \vec{t}_n) \geq  (\beta')^n \cdot D_p(\lat_n, (1-\eps/\delta)^{1/p}r_n)
	\; ,
	\end{equation}
	then we may replace the non-uniform reduction with a randomized reduction.
\end{corollary}
\begin{proof}
	By Theorem~\ref{thm:SAT_to_ESC}, it suffices to show a reduction from $\ESC_{\eta'}$ with $d = n/\eta'$ and $m  \in [n, C'n]$ for some constants $C' \geq 1$ and $\eta' \in (0,1)$. By Theorem~\ref{thm:sparsification_reduction}, it suffices to reduce this to $(A, G)\text{-}\CVP_{p,\gamma}$ on a rank $O(n)$ lattice with $A \leq 2^{-m} G$. By Theorem~\ref{thm:ESC_to_CVP}, such a reduction exists, but it requires a gadget a lattice of rank $O(n)$ satisfying Eq.~\eqref{eq:good_gadget} as auxiliary input.
	
	By Item~\ref{item:non-uniform} of Lemma~\ref{lem:exp_gap}, such a gadget exists for sufficiently large $n$ if Eq.~\eqref{eq:more_close_than_short_corollary} holds. Therefore, we get a non-uniform reduction that uses this gadget as advice.
	
	To prove the ``furthermore,'' it suffices to show that the additional assumptions are sufficient to make this gadget efficiently computable. Indeed, by Item~\ref{item:uniform} of Lemma~\ref{lem:exp_gap}, if Eq.~\eqref{eq:target_not_dumb_corollary} holds with $\delta = 1/\sqrt{\eta'}$, then we may take the gadget to be a scaling of $\lat_{n^\dagger}$, $\vec{t}_{n^\dagger}$, and $r_{n^\dagger}$ for an appropriate choice of $n^\dagger = O(n)$. The scaling and $n^\dagger$ are clearly efficiently computable. So, if both Eqs.~\eqref{eq:more_close_than_short_corollary} and~\eqref{eq:target_not_dumb_corollary} hold and the family is efficiently computable, then the advice is indeed efficiently computable, as needed.
\end{proof}

\subsection{Gap-ETH Hardness of \texorpdfstring{$\SVP_p$}{SVP\_p} for \texorpdfstring{$p > 2$}{p > 2}}
\label{sec:gapeth_lp}

In order to prove Gap-ETH hardness of $\SVP_p$, we will need the following lemma. The proof can be found in Section~\ref{sec:integer_points}.

\begin{lemma}
	\label{lem:integer_lattice_bound}
	For any constants $p > 2$ and $\delta\in (0,1)$, there exist (efficiently computable) constants $\beta > 1$, $t \in (0,1/2]$, $C_r > 0$, and $\eps \in (0,\delta)$, such that for any positive integer $n$,
	\[
		N_p(\Z^n, (1-\eps)^{1/p} \cdot r, \vec{t} ) \geq \beta^n \cdot N_p(\lat, r, \vec0)
		\; ,
	\]
	where $r := C_r n^{1/p}$ and $\vec{t} = (t,t,\ldots, t) \in \R^n$, and 
	\[
	N_p(\Z^n, (1-\eps)^{1/p}r, \vec{t}) \geq  \beta^n \cdot D_p(\Z^n, (1-\eps/\delta)^{1/p}r)
	 \; .
	\]
\end{lemma}

\begin{corollary}
	\label{cor:Gap3SAT_to_SVP}
	For any constants $p > 2$ and $C' \geq 1$,  there exists a constant $\gamma > 1$ such that there is an efficient (randomized) reduction from Gap-$3$-$\SAT_{\eta'}^{\le C'}$ on $n$ variables to $\SVP_{p, \gamma}$ on a lattice of rank and dimension $O(n)$
	
	In particular, for some constant $\gamma > 1$, there is no $2^{o(n)}$-time algorithm for $\SVP_{p, \gamma}$ unless (randomized) Gap-ETH is false.
\end{corollary}
\begin{proof}
	Simply combine Corollary~\ref{cor:gadget_to_hardness} with Lemma~\ref{lem:integer_lattice_bound}, with $\lat_n := \Z^n$, $r_n := C_r n^{1/p}$, and $\vec{t}_n := (t,t,\ldots, t) \in \R^n$.
\end{proof}

\subsection{Gap-ETH-hardness of \texorpdfstring{$\SVP_2$}{SVP\_2} under a certain assumption}
\label{sec:gapeth_l2}

We now show that, at least in the special case of the $\ell_2$ norm, we can simplify Corollary~\ref{cor:gadget_to_hardness} a bit further to get a relatively clean conditional hardness result. (See Theorem~\ref{thm:kissing_gives_hardness} below.) We focus on the $\ell_2$ norm because (1) we obtain the simplest statement in this case; and (2) hardness in the $\ell_2$ norm implies hardness in other norms. But, we mention in passing that qualitatively similar results hold for all $\ell_p$ norms.

\begin{lemma}
\label{lem:angle}
For any integer $n \ge 100$, let $\vec{u}\in \R^n$ be a fixed vector, and let $\vec{t} \in \R^n$ be a uniformly random unit vector in the $\ell_2$ norm. Then for any $0 < \theta_1 < \theta_2 < \pi$, the probability that the angle between $\vec{u}$ and $\vec{t}$ is between $\theta_1$ and $\theta_2$ is at least
 \[
 \int_{\theta_1}^{\theta_2} \sin^{n-2} \theta \, {\rm d} \theta \;.
 \]
\end{lemma}
\begin{proof}
The probability density function of the angle $\theta$ is proportional to $\sin^{n-2} \theta$~\cite{CFJ13}. So, it suffices to show that the constant of proportionality is at least one. I.e., it suffices to show that the integral of this from $0$ to $\pi$ is at most one. Indeed,
\begin{align*}
 \int_{0}^{\pi} \sin^{n-2} \theta \,{\rm d} \theta &= 2  \int_{0}^{\pi/2} \sin^{n-2} \theta \,{\rm d} \theta \\
 &= 2  \int_{0}^{2\pi/5} \sin^{n-2} \theta \,{\rm d} \theta + 2  \int_{2\pi/5}^{\pi/2} \sin^{n-2} \theta \,{\rm d} \theta \\
 &\le 2 \sin^{98}(2\pi/5) + 2 (\pi/2 - 2\pi/5) \\
 &\le 1\;.
 \end{align*}
 The result follows.
 \end{proof}

\begin{corollary}
	\label{cor:random_vector_close}
	For any $\eps, \delta \in (0,1/100)$ with  $\eps  \leq \sqrt{\delta}/10$ and vector $\vec{v} \in \R^n$ with $n \geq 100$ and $1 \leq \|\vec{v}\|_2^2 \leq 1+\delta$, we have
	\[
	\Pr[\|\vec{v} - \vec{t}\|_2^2 \leq 1-\eps] \geq \frac{\eps}{2\sqrt{\delta (1 + \delta)}} \cdot  \left(\frac{1 - 2\eps - \eps^2/\delta}{1+\delta}\right)^{n/2}
	\; ,
	\]
	where $\vec{t} \in \R^n$ is a vector of $\ell_2$ norm $\sqrt{\delta}$ chosen uniformly at random.
\end{corollary}
\begin{proof}
	Let $\theta$ be the angle between $\vec{v}$ and $\vec{t}$. We have
	\begin{align*}
	\Pr[\|\vec{v} - \vec{t} \|^2 \le 1 - \eps] &= \Pr[\|\vec{v}\|_2^2 + \|\vec{t}\|_2^2 - 2 \|\vec{v}\|_2 \|\vec{t}\|_2 \cos \theta \le 1 - \eps] \\
	&\ge \Pr\Big{[} \cos \theta \ge \frac{ 1+\delta + \delta - 1 + \eps}{2 \sqrt{\delta (1 + \delta)}}\Big{]} \\ 
	&\geq \Pr\big{[} \cos \theta \ge (2\delta + \eps)/(2\sqrt{\delta(1 + \delta)})\big{]}\\
	&\geq \int_{\arccos \left(\frac{\delta + \eps}{\sqrt{\delta(1 + \delta)}}\right)}^{\arccos\left(\frac{2\delta + \eps}{2\sqrt{\delta(1 + \delta)}}\right)} \sin^{n-2} \theta \,{\rm d} \theta\\
	&\geq \left(\arccos \left(\frac{\delta + \eps}{\sqrt{\delta(1 + \delta)}}\right)- \arccos \left(\frac{2\delta + \eps}{2\sqrt{\delta(1 + \delta)}}\right) \right) \cdot \left(\frac{1 - 2\eps - \eps^2/\delta}{1+\delta}\right)^{n/2}\\
	&\geq\frac{\eps}{2\sqrt{\delta (1 + \delta)}} \cdot  \left(\frac{1 - 2\eps - \eps^2/\delta}{1+\delta}\right)^{n/2}
	\; ,
	\end{align*}
	where the first inequality uses the fact that $x - a/x$ is an increasing function of $x$ if $a >0$, the second-to-last inequality uses the fact that $\sin(\arccos(x)) = \sqrt{1-x^2}$, and the last inequality uses the fact that $\frac{\partial}{\partial x} \arccos(x) = -(1-x^2)^{-1/2} \leq -1$.
\end{proof}

\begin{corollary}
	\label{cor:kissing_gives_local_density}
	For any lattice $\lat \subset \R^n$ with $n \geq 100$, $\eps, \delta \in (0,1/100)$ with  $\eps  \leq \sqrt{\delta}/10$, radius $r > 0$, and target $\vec{t} \in \R^n$, there exists a vector $\vec{t}' \in \R^n$ such that
	\[
	N_2(\lat, \sqrt{1-\eps} \cdot r, \vec{t}') \geq \frac{\eps}{2\sqrt{\delta (1 + \delta)}} \cdot  \left(\frac{1 - 2\eps - \eps^2/\delta}{1+\delta}\right)^{n/2}
 \cdot N_2(\lat, \sqrt{1+\delta} \cdot r, \vec{t})
	\; .
	\]
\end{corollary}
\begin{proof}
	Take $\vec{t}' \in \R^n$ to be a uniformly random vector that is at $\ell_2$ distance $\sqrt{\delta} \cdot r$ away from $\vec{t}$ as in Corollary~\ref{cor:random_vector_close}. Then, by the corollary, the expectation of $N_2(\lat, \sqrt{1-\eps} \cdot r, \vec{t}')$ is at least the right-hand side of the above inequality, and the result follows.
\end{proof}

We now show that a family of lattices with ``surprisingly many'' points in a ball is enough to instantiate Corollary~\ref{cor:gadget_to_hardness}. In more detail, notice that we expect $N_2(\lat, r, \vec{t})$ to ``typically'' be proportional $r^n$, and in particular, we expect $N_2(\lat, r', \vec{t})  \lesssim (r'/r)^n \cdot N_2(\lat, r, \vec0)$ for $r' > r$. We show that, in order to instantiate Corollary~\ref{cor:gadget_to_hardness}, it suffices to find a family of lattices, radii, and targets such $N_2(\lat, r', \vec{t})  \geq \beta^n \cdot (r'/r)^n \cdot N_2(\lat, r, \vec0)$ for some constant $\beta > 1$ and $r' \leq O(r)$. For example, by taking $r = \lambda_1^{(2)}(\lat) - \eps$, $r' = \lambda_1^{(2)}(\lat)$, and $\vec{t} = \vec0$, we see that it would suffice to find a family of lattices with exponentially many non-zero vectors of minimal length in the $\ell_2$ norm---i.e., a family with exponentially large kissing number.

\begin{theorem}
	\label{thm:kissing_gives_hardness}
	Suppose that for some constants $\beta,\alpha > 1$, the following holds. For every sufficiently large integer $n$, there exists a lattice $\lat_n \subset \R^n$, radii $0 < r_n \leq r_n' \leq \alpha r_n$, and a target $\vec{t}_n \in \R^n$  such that
	\begin{equation}
	\label{eq:kissing}
		N_2(\lat_n, r_n', \vec{t}_n) \geq  \beta^n \cdot (r_n'/r_n)^n \cdot N_2(\lat_n, r_n, \vec0)
		\; .
	\end{equation}
	Then for any constant $C \geq 1$ and $\eta \in (0,1)$, there is an efficient (non-uniform) reduction from from Gap-$3$-$\SAT_{\eta}^{\le C}$ on $n$ variables to $\SVP_{2, \gamma}$ on a lattice of rank and dimension $O(n)$, for some constant $\gamma > 1$.
	
	In particular, if such a family of lattices and radii exists, then for some constant $\gamma > 1$, there is no $2^{o(n)}$-time algorithm for $\SVP_{2,\gamma}$ unless (non-uniform) Gap-ETH is false.
\end{theorem}
\begin{proof}
	We first prove the theorem under the assumption that $r_n'/r_n < 1+1/400$.
	We may assume without loss of generality that $r_n'/r_n > \beta^{1/2}$. (Otherwise, we may replace $r_n'$ with $r_n \cdot \beta^{1/2}$ and $\beta$ with $\beta^{1/2}$. Then, clearly $r_n'/r_n > \beta^{1/2}$ and Eq.~\eqref{eq:kissing} still holds.)
	
	Let $s_n := r_n \cdot \sqrt{1-\eps'}/\sqrt{1-\eps}$ for some small constants $0 < \eps' < \eps < 1/100$ to be chosen later. By Corollary~\ref{cor:gadget_to_hardness}, it suffices to find a $\vec{t}_n' \in \R^n$ such that 
		\begin{equation}
		\label{eq:more_close_than_short_one_more_time}
		N_2(\lat_n, \sqrt{1-\eps'} \cdot r_n, \vec{t}_n') = N_2(\lat_n, \sqrt{1-\eps} \cdot  s_n, \vec{t}_n') \geq (\beta')^n \cdot N_2(\lat_n, r_n, \vec{0})
		\; ,
		\end{equation}
		for some constant $\beta' > 1$ and sufficiently large $n$. 		
		Let $\delta := (r_n'/s_n)^2 - 1 \in (\beta^{1/2}-1,1/100)$. By Corollary~\ref{cor:kissing_gives_local_density}, as long as $\eps < \sqrt{\beta^{1/2} - 1}/10$, we see that there exists 
		a $\vec{t}_n' \in \R^n$ with
		\begin{align*}
		N_2(\lat_n, \sqrt{1-\eps} \cdot s_n, \vec{t}_n')
		&\geq \frac{\eps}{2\sqrt{\delta (1 + \delta)}} \cdot  \left(\frac{1 - 2\eps - \eps^2/\delta}{1+\delta}\right)^{n/2}
		\cdot N_2(\lat_n, \sqrt{1+\delta} \cdot s_n, \vec{t}_n)\\
		&= \frac{\eps}{2\sqrt{\delta (1 + \delta)}} \cdot  \left(\frac{1 - 2\eps - \eps^2/\delta}{1+\delta}\right)^{n/2}
		\cdot N_2(\lat_n, r_n', \vec{t}_n)\\
		&\geq \frac{\eps}{2\sqrt{\delta (1 + \delta)}} \cdot  \left(\frac{1 - 2\eps - \eps^2/\delta}{1+\delta}\right)^{n/2} \cdot (r_n'/r_n)^n \cdot \beta^n
		\cdot N_2(\lat_n, r_n, \vec0) \\
		&> \frac{\eps}{2\sqrt{\delta (1 + \delta)}} \cdot  (1 - 2\eps - \eps^2/\delta)^{n/2}\cdot \beta^n
		\cdot N_2(\lat_n, r_n, \vec0)
		\; .
		\end{align*}
	The result follows by taking $\eps$ and $\beta'$ to be a small enough constant that $(1 - 2\eps - \eps^2/\delta)^{1/2}\cdot \beta > \beta' > 1$.
	
	Now, suppose $r_n'/r_n \geq 1+1/400$. We claim that there exists some $0 < s_n \leq s_n' < (1+1/400) \cdot s_n$ such that 
	\begin{equation}
	\label{eq:yet_another_more_close_than_short}
	D_2(\lat_n, s_n') \geq  (\beta')^n \cdot (s_n'/s_n)^n \cdot N_2(\lat_n, s_n, \vec0)
	\end{equation}
	for some constant $\beta' > 1$ depending only on $\beta$ and $\alpha$. Clearly, this suffices to prove the result. 
	
	To that end, for $i = 0,\ldots, \ell := \ceil{2000 \log(r_n'/r_n)}$, let $s_n^{(i)} := (r_n'/r_n)^{i/\ell} \cdot r_n$, $D_n^{(i)} := D_2(\lat_n, s_n^{(i)})$, and $N_n^{(i)} := N_2(\lat_n, s_n^{(i)}, \vec0)$.  We take to be the constant $\beta' := \beta^{1/\ceil{2000 \beta \log \alpha}} \leq \beta^{1/\ell}$.
	We claim that it suffices to find an index $i$ such that $D_n^{(i+1)}/N_n^{(i)}  \geq (r_n'/r_n)^{n/\ell} \beta^{n/\ell}$. Indeed, if such an index exists, then we can take $s_n' := s_n^{(i+1)}$ and $s_n := s_n^{(i)}$. Clearly, $s_n'/s_n = (r_n'/r_n)^{1/\ell} \leq 1+1/400$, and Eq.~\eqref{eq:yet_another_more_close_than_short} is indeed satisfied.
	
	It remains to find such an index $i$.
	By assumption, we have $D_n^{(\ell)}/N_n^{(0)} \geq (r_n'/r_n)^n \cdot \beta^n$, and by definition, we have $D_n^{(i)}  \geq N_n^{(i)}$.   If there exists an index $i$ such that $D_n^{(i+1)}/D_n^{(i)} \geq (r_n'/r_n)^{n/\ell} \beta^{n/\ell}$, then we are done, since $N_n^{(i)} \leq D_n^{(i)}$. Otherwise, we have 
	\[
	D_n^{(1)} \geq (r_n'/r_n)^{-(\ell - 1)n/\ell} \cdot \beta^{-(\ell -1) n/\ell} D_n^{(\ell)} \geq (r_n'/r_n)^{n/\ell} \cdot \beta^{n/\ell} N_n^{(0)}
	\; ,
	\]
	as needed.
	
\end{proof}

\begin{corollary}
	\label{cor:ellp_hard}
	Suppose that for some constants $\beta,\alpha > 1$, the following holds. For every sufficiently large integer $n$, there exists a lattice $\lat_n \subset \R^n$, radii $0 < r_n \leq r_n' \leq \alpha r_n$, and a target $\vec{t}_n \in \R^n$  such that
	\[
	N_2(\lat_n, r_n', \vec{t}_n) \geq  \beta^n \cdot (r_n'/r_n)^n \cdot N_2(\lat_n, r_n, \vec0)
	\; .
	\]
	Then for any constants $C \geq 1$, $\eta \in (0,1)$, and $p \in [1,2]$, there is an efficient (non-uniform) reduction from from Gap-$3$-$\SAT_{\eta}^{\le C}$ on $n$ variables to $\SVP_{p, \gamma}$ on a lattice of rank and dimension $O(n)$, for some constant $\gamma > 1$.
	
	In particular, if such a family of lattices and radii exists, then for each constant $p \in [1,2]$ there exists a constant $\gamma_p > 1$, such that there is no $2^{o(n)}$-time algorithm for $\SVP_{p,\gamma_p}$ unless (non-uniform) Gap-ETH is false.
\end{corollary}
\begin{proof}
	Combine Thereom~\ref{thm:kissing_gives_hardness} with Corollary~\ref{cor:embedding_reduction}.
\end{proof}

	\section{On the number of integer points in an \texorpdfstring{$\ell_p$}{ell p} ball}
	\label{sec:integer_points}

In this section, we prove Lemma~\ref{lem:integer_lattice_bound} by studying the function $N_p(\Z^n, r, \vec{t})$. The results in this section were originally proven by Mazo and Odlyzko~\cite{MO90} for $p = 2$ and Elkies, Odlyzko, and Rush~\cite{EOR91} for arbitrary $p$. In particular, the main theorem of this section, Theorem~\ref{thm:theta_gives_good_approx}, appeared in~\cite{MO90} for $p = 2$ and in~\cite{EOR91} for arbitrary $p$ (and even in a more general setting), and (a variant of) Lemma~\ref{lem:integer_lattice_bound} already appeared in~\cite{EOR91}. Our proof mostly follows that of Mazo and Odlyzko.

	\subsection{Approximation by \texorpdfstring{$\Theta_p$}{Theta}}
	
	We extend our definition of $\Theta_p(\tau)$ to ``shifted sums'' as follows. For $1 \leq p < \infty$, $\tau > 0$, and $t \in \R$, let
	\[
	\Theta_p(\tau; t) := \sum_{z \in \Z} \exp(-\tau |z - t|^p)
	\; .
	\]
	We then further extend this definition to vectors $\vec{t} = (t_1,\ldots, t_n) \in \R^n$ by
	\[
	\Theta_p(\tau; \vec{t}) := \prod_{i=1}^n \Theta_p(\tau; t_i)
	\]
		We will often assume without loss of generality that $t \in [0,1/2]$ and $\vec{t} \in [0,1/2]^n$.

	We have
	\[
	\Theta_p(\tau; \vec{t}) = \sum_{\vec{z} \in \Z^n} \exp(-\tau \|\vec{z} - \vec{t}\|_p^p)
	\; .
	\]
	It follows that for any radius $r > 0$ and any $\tau > 0$,
	\begin{equation}
	\label{eq:basic_upper_bound_theta}
	N_p(\Z^n, r, \vec{t}) \leq \exp(\tau r^p) \Theta_p(\tau; \vec{t})
	\; .
	\end{equation}
	We wish to show that this inequality is quite tight if $\tau$ satisfies $\mu_p(\tau; \vec{t}) = r^p$, where
	\[
	\mu_p(\tau; \vec{t}) := \sum_{i=1}^n \expect_{X \sim D_p(\tau; t_i)}[|X|^p]
	\; ,
	\]
	and $D_p(\tau; t)$ is the probability distribution over $\Z - t$ that assigns to each $x \in \Z -t$ probability $\exp(-\tau |x|^p)/\Theta_p(\tau; t)$.\footnote{It is easy to see that there exists a $\tau$ satisfying $\mu_p(\tau; \vec{t}) = r^p$ if and only if there is a lattice point inside the \emph{open} $\ell_p$ ball of radius $r$ around $\vec{t}$. So, we do not lose much by assuming that $r^p = \mu_p(\tau; \vec{t})$ for some $\tau > 0$.}
	Indeed, the main theorem that we prove in this section is the following (which, again, was originally proven in~\cite{MO90,EOR91}).
	
	\begin{theorem}
		\label{thm:theta_gives_good_approx}
		For any constants $p \geq 1$ and $\tau > 0$, there is another constant $C^* > 0$ such that
		\[
		\exp(\tau \mu_p(\tau; \vec{t}) -C^*\sqrt{n}) \cdot  \Theta_p(\tau; \vec{t}) \leq N_p(\Z^n, \mu_p(\tau; \vec{t})^{1/p}, \vec{t}) \leq \exp(\tau \mu_p(\tau; \vec{t}) ) \cdot  \Theta_p(\tau; \vec{t})
		\; ,
		\]
		for any $\vec{t} \in \R^n$ and any positive integer $n$.
	\end{theorem}

	\begin{lemma}
		\label{lem:theta_derivatives}
		For any $1 \leq p < \infty$, $\tau > 0$, and $t \in \R$,
		\[
		\frac{\partial }{\partial \tau } \log \Theta_p(\tau; t) = -\mu_p(\tau; t)  < 0
		\; ,
		\]
		and 
		\[
		\frac{\partial^2 }{\partial \tau^2} \log \Theta_p(\tau; t) = \expect_{X \sim D_p(\tau; t) }[|X|^{2p}] - \mu_p(\tau; t)^2 > 0
		\; .
		\]
	\end{lemma}
	\begin{proof}
		We have
		\begin{align*}
		\frac{\partial }{\partial \tau } \log \Theta_p(\tau; t)  &= \frac{1}{\Theta(\tau; t)} \cdot \sum_{z \in \Z} \frac{\partial }{\partial \tau} \exp(-\tau |z-t|^p)\\
		&= -\frac{1}{\Theta(\tau; t)} \cdot \sum_{z \in \Z} |z-t|^p \exp(-\tau |z-t|^p)\\
		&= -\mu_p(\tau; t)
		\; .
		\end{align*}
		The second derivative follows from a similar computation.
	\end{proof}
	
	To use $\Theta_p(\tau; \vec{t})$ to get a lower bound on $N_p(\Z^n, r, \vec{t})$, we wish to ``isolate'' the contribution of vectors of length roughly $r$ to the sum $\Theta_p(\tau; \vec{t})$. To do so, we define for $\delta > 0$ the (rather unnatural) function
	\[
	H_p(\tau, \delta; \vec{t}) := \Theta_p(\tau + \delta; \vec{t}) - \exp(-\delta \mu_p(\tau; \vec{t})) \Theta_p(\tau; \vec{t}) - \exp(\delta \mu_p(\tau + 2\delta; \vec{t})) \Theta_p(\tau + 2\delta; \vec{t})
	\; .
	\]
	The following lemma shows why we are interested in $H_p$.
	
	\begin{lemma}
		\label{lem:H_gives_lower_bound}
		For any $1 \leq p < \infty$, $\tau > 0$, $\delta > 0$, and $\vec{t} \in [0,1/2]^n$,
		\[
		N_p(\Z^n, \mu_p(\tau; \vec{t})^{1/p}, \vec{t}) \geq \exp(\tau \mu_p(\tau + 2\delta; \vec{t}) ) H_p(\tau, \delta; \vec{t}) 
		\; .
		\]
	\end{lemma}
	\begin{proof}
		We can write
		\[
		H_p(\tau, \delta; \vec{t}) = \sum_{\vec{z} \in \Z^n} \exp(-(\tau + \delta) \|\vec{z} - \vec{t}\|_p^p) \cdot \big(1 - f_1(\vec{z}) - f_2(\vec{z})\big)
		\; ,
		\]
		where $f_1(\vec{z}) := \exp(\delta (\|\vec{z} - \vec{t}\|_p^p - \mu_p(\tau; \vec{t})))$, and $f_2(\vec{z}) :=  \exp(\delta (\mu_p(\tau + 2\delta; \vec{t}) - \|\vec{z} - \vec{t}\|_p^p))$. In particular, the summand is negative if $\|\vec{z} - \vec{t}\|_p^p \geq \mu_p(\tau; \vec{t})$ or $\|\vec{z} - \vec{t}\|_p^p \leq \mu_p(\tau + 2\delta; \vec{t})$. Therefore, there are at most $N_p(\Z^n, \mu_p(\tau; \vec{t})^{1/p}, \vec{t})$ positive summands, and each has magnitude at most 
		\[
		\exp(-(\tau + \delta)\mu_p(\tau + 2\delta; \vec{t})) \leq \exp(-\tau \mu_p(\tau + 2\delta; \vec{t}))
		\; .
		\]
		The result follows.
	\end{proof}
	
	\begin{proof}[Proof of Theorem~\ref{thm:theta_gives_good_approx}]
		We may assume without loss of generality that $\vec{t} \in [0,1/2]^n$.
		Let $\delta > 0$ with $\delta = O(1)$ to be chosen later. By taking the Taylor expansion of $\log \Theta_p(\tau; \vec{t})$ around $\tau$, we see by Lemma~\ref{lem:theta_derivatives} that
		\[
		\log \Theta_p(\tau; \vec{t}) - \log \Theta_p(\tau + \delta; \vec{t}) \leq \delta \mu_p(\tau; \vec{t}) - \frac{\delta^2}{2} \inf_{\tau \leq \tau' \leq \tau + \delta} \sum_{i=1}^n V(\tau'; t_i)
		\; ,
		\]
		where 
		\[
		V(\tau'; t_i) :=  \expect_{X \sim D_p(\tau'; t_i) }[|X|^{2p}] - \mu_p(\tau'; t_i)^2
		\; .
		\]
		Notice that $V(\tau'; t_i)$ is a continuous positive function. Therefore, it has a positive lower bound on the compact set $\tau \leq \tau' \leq \tau + 2\delta$ and $0 \leq t_i \leq 1/2$. (We have deliberately taken the upper limit on $\tau'$ to be $\tau + 2\delta$ rather than just $\tau + \delta$.) Let $C_{\min} > 0$ be such a bound. We therefore have
		\[
		\log \Theta_p(\tau; \vec{t}) - \log \Theta_p(\tau + \delta; \vec{t}) \leq \delta \mu_p(\tau; \vec{t}) - C_{\min}n \delta^2 /2
		\; .
		\]
		By an essentially identical argument, we have, 
		\[
		\log \Theta_p(\tau + 2\delta; \vec{t})-\log \Theta_p(\tau + \delta; \vec{t})  \leq -\delta \mu_p(\tau + 2\delta; \vec{t}) - C_{\min} n\delta^2 /2
		\; .
		\]
		It follows that
		\[
		H_p(\tau, \delta; \vec{t}) \geq \Theta_p(\tau + \delta; \vec{t}) \cdot (1-2\exp(-C_{\min}n\delta^2 /2))
		\; .
		\]
		
		Plugging this in to Lemma~\ref{lem:H_gives_lower_bound}, we see that
		\begin{align*}
		N_p(\Z^n, \mu_p(\tau; \vec{t})^{1/p}, \vec{t}) \geq \exp(\tau \mu_p(\tau + 2\delta; \vec{t}) ) \Theta_p(\tau + \delta; \vec{t}) \cdot (1-2\exp(-C_{\min}n\delta^2 /2))
		\; .
		\end{align*}
		By a similar argument, we see that there is some constant $C_{\max} > 0$ such that $\mu_p(\tau; t_i) - \mu_p(\tau + 2\delta; t_i) \leq C_{\max} \delta$ \emph{and} $\log \Theta_p(\tau; t_i) - \log \Theta_p(\tau + \delta; t_i) \leq C_{\max} \delta$. Therefore,
		\[
		N_p(\Z^n, \mu_p(\tau; \vec{t})^{1/p}, \vec{t})  \geq \exp(\tau \mu_p(\tau; \vec{t}) ) \Theta_p(\tau; \vec{t}) \cdot \exp(-C_{\max}\delta n(\tau + 1)) (1-2\exp(-C_{\min}n\delta^2 /2))
		\; .
		\]
		The result follows by taking $\delta = C^\dagger/\sqrt{n}$ for a sufficiently large constant $C^\dagger > 0$.
	\end{proof}
	
	\subsection{Dense shifted balls and the proof of Lemma~\ref{lem:integer_lattice_bound}}
	
	We now wish to show that for $p > 2$, there exist shifts $\vec{t} \in \R^n$ such that $N_p(\Z^n, r, \vec{t})$ is exponentially larger than $N_p(\Z^n, r, \vec0)$. (Again, this result is not original to us, as it was already proven in~\cite{EOR91}.) By Theorem~\ref{thm:theta_gives_good_approx}, it suffices to show that there exists $\tau > 0$ and $t \in (0,1/2]$ such that $\Theta_p(\tau; t) > \Theta_p(\tau; 0)$.
	
	\begin{lemma}
		\label{lem:local_minimum_theta}
	For any $p > 2$ and $\tau \geq 1-1/p$, there exists a $t \in (0,1/2]$ such that 
	\[
	\Theta_p(\tau; t) > \Theta_p(\tau; 0)
	\; .
	\]
	\end{lemma}
	\begin{proof}
		We have $\frac{\partial }{\partial t}  \Theta_p(\tau; t) |_{t = 0} = 0$,
		and
		\begin{equation}
		\label{eq:second_deriv_Theta_t}
		\frac{\partial^2 }{\partial t^2}  \Theta_p(\tau; t) \big|_{t = 0} 
		= p\tau \sum_{z \in \Z} \exp(-\tau |z|^p) |z|^{p-2}(p \tau |z|^p-(p-1))
		\; .
		\end{equation}
		(Here, we have used the fact that $\exp(-\tau |t|^p)$ is twice differentiable at $t = 0$ for $p > 2$, with first and second derivative both zero. Notice that this is false for $p \leq 2$.) 
		For $\tau \geq 1-1/p$, all of the summands are non-negative, so that $0$ is a local minimum of the function $t \mapsto \Theta_p(\tau; t)$. Therefore, for sufficiently small $t > 0$, $\Theta_p(\tau; t) > \Theta_p(\tau; 0)$, as needed.
	\end{proof}
	
	\begin{proof}[Proof of Lemma~\ref{lem:integer_lattice_bound}]
		Choose $t \in [0,1/2]$ to maximize $\Theta_p(1; t)$. (Since $\Theta_p(1;t)$ is a continuous function and $[0,1/2]$ is a compact set, such a maximizer is guaranteed to exist.)  By Lemma~\ref{lem:local_minimum_theta}, we must have $\Theta_p(1; t) > \Theta_p(1; 0)$, and in particular $t > 0$. Let $\eps \in (0,\delta)$ be a constant to be chosen later. Let $C_r := \mu_p(1;t)^{1/p}/(1-\eps)^{1/p}$.
		
		Let $\eps \in (0,\delta)$ be a constant to be chosen later, and let $r := C_r n^{1/p}$ and $\vec{t} := (t,t,\ldots, t) \in \R^n$. Notice that $(1-\eps) r^p = \mu_p(1; \vec{t})$. By Theorem~\ref{thm:theta_gives_good_approx}, we have
		\begin{equation}
		\label{eq:t_bound}
		N_p(\Z^{n}, (1-\eps)^{1/p} \cdot r, \vec{t}) \geq \exp(-C^*\sqrt{n}) \cdot \exp((1-\eps) r^p) \cdot  \Theta_p(1; t)^{n}
		\; 
		\end{equation}
		for some constant $C^* > 0$.
		
		We have
		\[
		N_p(\Z^{n}, r, \vec{0}) \leq \exp(r^p) \Theta_p(1; 0)^{ n}
		\; .
		\]
		Let $\alpha := \Theta_p(1; t)/\Theta_p(1; 0) > 1$.
		Then, combining the above with Eq.~\eqref{eq:t_bound}, we see that
		\[
		\frac{N_p(\Z^{n}, (1-\eps)^{1/p} r, \vec{t})}{N_p(\Z^{n}, r, \vec{0})}
		\geq \alpha^{n} \cdot \exp(-\eps C_r^p n)
		\; .
		\]
		We therefore take $\eps \in (0,\delta)$ to be any constant small enough to make $\alpha \exp(-\eps C_r^p) > 1$.
		
		Now, for any $\vec{t}' = (t_1',\ldots, t_{n}') \in \R^{n}$, we may repeat the above argument to show that
		\[
		\frac{N_p(\Z^{n}, (1-\eps)^{1/p} r, \vec{t})}{N_p(\Z^{n}, (1-\eps/\delta) r, \vec{t}')} \geq \exp( \eps C_r^p n \cdot (1/\delta - 1)) \cdot \Theta_p(1; t)^n/\Theta_p(1;\vec{t}')
		\]
		Recall that, by definition, $\Theta_p(1; \vec{t}') = \prod_i \Theta_p(1; t_i') \leq \Theta_p(1; t)^{n}$, where the last inequality follows from our choice of $t$.  Therefore,
		\[
		\frac{N_p(\Z^{n}, (1-\eps)^{1/p} r, \vec{t})}{D_p(\Z^n,(1-\eps/\delta) r)} \geq \exp( \eps C_r^p n \cdot (1/\delta - 1)) 
		\; .
		\]
		The result follows by taking $\beta := \min\{ \alpha \exp(-\eps C_r^p),\ \exp(\eps C_r^p (1/\delta - 1))\} > 1$.
	\end{proof}

\bibliographystyle{alpha}

\begin{thebibliography}{WLTB11}

\bibitem[ADPS16]{new_hope}
Erdem Alkim, L{\'e}o Ducas, Thomas P{\"o}ppelmann, and Peter Schwabe.
\newblock Post-quantum key exchange --- {A} new hope.
\newblock In {\em USENIX Security Symposium}, 2016.

\bibitem[ADRS15]{ADRS15}
Divesh Aggarwal, Daniel Dadush, Oded Regev, and Noah Stephens{-}Davidowitz.
\newblock Solving the {S}hortest {V}ector {P}roblem in $2^n$ time via discrete
  {G}aussian sampling.
\newblock In {\em STOC}, 2015.

\bibitem[AJ08]{AJ08}
Vikraman Arvind and Pushkar~S Joglekar.
\newblock Some sieving algorithms for lattice problems.
\newblock In {\em FSTTCS}, pages 25--36, 2008.

\bibitem[Ajt98]{Ajtai-SVP-hard}
Mikl\'{o}s Ajtai.
\newblock The {S}hortest {V}ector {P}roblem in {L2} is {NP}-hard for randomized
  reductions.
\newblock In {\em STOC}, 1998.

\bibitem[Ajt04]{Ajtai96}
Mikl{\'o}s Ajtai.
\newblock Generating hard instances of lattice problems.
\newblock In {\em Complexity of computations and proofs}, volume~13 of {\em
  Quad. Mat.}, pages 1--32. Dept. Math., Seconda Univ. Napoli, Caserta, 2004.
\newblock Preliminary version in STOC'96.

\bibitem[AKS01]{AKS01}
Mikl\'{o}s Ajtai, Ravi Kumar, and D.~Sivakumar.
\newblock A sieve algorithm for the shortest lattice vector problem.
\newblock In {\em STOC}, pages 601--610, 2001.

\bibitem[Alo97]{Alon97}
Noga Alon.
\newblock Packings with large minimum kissing numbers.
\newblock {\em Discrete Mathematics}, 175(1):249 -- 251, 1997.

\bibitem[AR05]{AharonovR04}
Dorit Aharonov and Oded Regev.
\newblock Lattice problems in {NP} intersect {coNP}.
\newblock {\em Journal of the ACM}, 52(5):749--765, 2005.
\newblock Preliminary version in FOCS'04.

\bibitem[AS17]{AS17}
Divesh Aggarwal and Noah Stephens{-}Davidowitz.
\newblock Just take the average! an embarrassingly simple $2^n$-time algorithm
  for {SVP} (and {CVP}), 2017.
\newblock \url{http://arxiv.org/abs/1709.01535}.

\bibitem[BCD{\etalchar{+}}16]{frodo}
Joppe~W. Bos, Craig Costello, L{\'{e}}o Ducas, Ilya Mironov, Michael Naehrig,
  Valeria Nikolaenko, Ananth Raghunathan, and Douglas Stebila.
\newblock Frodo: Take off the ring! {P}ractical, quantum-secure key exchange
  from {LWE}.
\newblock In {\em CCS}, 2016.

\bibitem[BDGL16]{BDGL16}
Anja Becker, L{\'e}o Ducas, Nicolas Gama, and Thijs Laarhoven.
\newblock New directions in nearest neighbor searching with applications to
  lattice sieving.
\newblock In {\em SODA}, 2016.

\bibitem[BGS17]{BGS17}
Huck Bennett, Alexander Golovnev, and Noah Stephens{-}Davidowitz.
\newblock On the quantitative hardness of {CVP}.
\newblock In {\em FOCS}, 2017.

\bibitem[BN09]{BN09}
Johannes Bl{\"o}mer and Stefanie Naewe.
\newblock Sampling methods for shortest vectors, closest vectors and successive
  minima.
\newblock {\em Theoret. Comput. Sci.}, 410(18):1648--1665, 2009.

\bibitem[CFJ13]{CFJ13}
Tony Cai, Jianqing Fan, and Tiefeng Jiang.
\newblock Distributions of angles in random packing on spheres.
\newblock {\em The Journal of Machine Learning Research}, 14(1):1837--1864,
  2013.

\bibitem[CN98]{CN98}
J-Y Cai and Ajay Nerurkar.
\newblock Approximating the {SVP} to within a factor $(1+1/\dim^\eps)$ is
  {NP}-hard under randomized conditions.
\newblock In {\em CCC}. IEEE, 1998.

\bibitem[CS98]{ConwaySloaneBook98}
John Conway and Neil~J.A. Sloane.
\newblock {\em Sphere Packings, Lattices and Groups}.
\newblock Springer New York, 1998.

\bibitem[Din16]{journals/eccc/Dinur16}
Irit Dinur.
\newblock Mildly exponential reduction from gap {3SAT} to polynomial-gap
  label-cover.
\newblock {\em Electronic Colloquium on Computational Complexity {(ECCC)}},
  23:128, 2016.

\bibitem[DPV11]{DPV11}
Daniel Dadush, Chris Peikert, and Santosh Vempala.
\newblock Enumerative lattice algorithms in any norm via {M}-ellipsoid
  coverings.
\newblock In {\em FOCS}, 2011.

\bibitem[EOR91]{EOR91}
N.~D. Elkies, A.~M. Odlyzko, and J.~A. Rush.
\newblock On the packing densities of superballs and other bodies.
\newblock {\em Inventiones mathematicae}, 105(1):613--639, Dec 1991.

\bibitem[FLM76]{FLM76}
T.~Figiel, J.~Lindenstrauss, and V.~D. Milman.
\newblock The dimension of almost spherical sections of convex bodies.
\newblock {\em Bull. Amer. Math. Soc.}, 82(4):575--578, 07 1976.

\bibitem[GN08]{GN08}
Nicolas Gama and Phong~Q. Nguyen.
\newblock Finding short lattice vectors within {M}ordell's inequality.
\newblock In {\em {STOC}}, 2008.

\bibitem[GPV08]{GPV08}
Craig Gentry, Chris Peikert, and Vinod Vaikuntanathan.
\newblock Trapdoors for hard lattices and new cryptographic constructions.
\newblock In {\em STOC}, 2008.

\bibitem[HR12]{HRsvp}
Ishay Haviv and Oded Regev.
\newblock Tensor-based hardness of the {S}hortest {V}ector {P}roblem to within
  almost polynomial factors.
\newblock {\em Theory of Computing}, 8(23):513--531, 2012.
\newblock Preliminary version in STOC'07.

\bibitem[IP99]{IP1999}
Russell Impagliazzo and Ramamohan Paturi.
\newblock On the complexity of $k$-{SAT}.
\newblock In {\em CCC}, pages 237--240, 1999.

\bibitem[JS98]{JS98}
Antoine Joux and Jacques Stern.
\newblock Lattice reduction: A toolbox for the cryptanalyst.
\newblock {\em Journal of Cryptology}, 11(3):161--185, 1998.

\bibitem[Kan87]{Kannan87}
Ravi Kannan.
\newblock Minkowski's convex body theorem and integer programming.
\newblock {\em Math. Oper. Res.}, 12(3):415--440, 1987.

\bibitem[Kho05]{Khot05svp}
Subhash Khot.
\newblock Hardness of approximating the {S}hortest {V}ector {P}roblem in
  lattices.
\newblock {\em Journal of the ACM}, 52(5):789--808, September 2005.
\newblock Preliminary version in FOCS'04.

\bibitem[Laa15]{Laarhoven2015}
Thijs Laarhoven.
\newblock Sieving for shortest vectors in lattices using angular
  locality-sensitive hashing.
\newblock In {\em CRYPTO}, 2015.

\bibitem[Len83]{Lenstra83}
H.~W. Lenstra, Jr.
\newblock Integer programming with a fixed number of variables.
\newblock {\em Math. Oper. Res.}, 8(4):538--548, 1983.

\bibitem[LLL82]{LLL82}
A.~K. Lenstra, H.~W. Lenstra, Jr., and L.~Lov{\'a}sz.
\newblock Factoring polynomials with rational coefficients.
\newblock {\em Math. Ann.}, 261(4):515--534, 1982.

\bibitem[LWXZ11]{LWXZ11}
Mingjie Liu, Xiaoyun Wang, Guangwu Xu, and Xuexin Zheng.
\newblock Shortest lattice vectors in the presence of gaps.
\newblock {\em IACR Cryptology ePrint Archive}, 2011:139, 2011.

\bibitem[Man17]{PasinPrivate}
Pasin Manurangsi, 2017.
\newblock Personal communication.

\bibitem[MG02]{MicciancioBook}
Daniele Micciancio and Shafi Goldwasser.
\newblock {\em Complexity of Lattice Problems: a cryptographic perspective},
  volume 671 of {\em The Kluwer International Series in Engineering and
  Computer Science}.
\newblock Kluwer Academic Publishers, Boston, Massachusetts, March 2002.

\bibitem[Mic01]{Mic01svp}
Daniele Micciancio.
\newblock The {S}hortest {V}ector {P}roblem is {NP}-hard to approximate to
  within some constant.
\newblock {\em SIAM Journal on Computing}, 30(6):2008--2035, March 2001.
\newblock Preliminary version in FOCS 1998.

\bibitem[Mic12]{Mic12}
Daniele Micciancio.
\newblock Inapproximability of the {S}hortest {V}ector {P}roblem: Toward a
  deterministic reduction.
\newblock {\em Theory of Computing}, 8(22):487--512, 2012.

\bibitem[MO90]{MO90}
J.~E. Mazo and A.~M. Odlyzko.
\newblock Lattice points in high-dimensional spheres.
\newblock {\em Monatsh. Math.}, 110(1):47--61, 1990.

\bibitem[MR16]{MR17}
Pasin Manurangsi and Prasad Raghavendra.
\newblock A birthday repetition theorem and complexity of approximating dense
  {CSP}s.
\newblock {\em arXiv preprint arXiv:1607.02986}, 2016.

\bibitem[MV10]{MV10}
Daniele Micciancio and Panagiotis Voulgaris.
\newblock Faster exponential time algorithms for the {S}hortest {V}ector
  {P}roblem.
\newblock In {\em SODA}, 2010.

\bibitem[MW16]{MW16}
Daniele Micciancio and Michael Walter.
\newblock Practical, predictable lattice basis reduction.
\newblock In {\em Eurocrypt}, 2016.

\bibitem[NIS16]{NIST_quantum}
{NIST} post-quantum standardization call for proposals.
\newblock
  \url{http://csrc.nist.gov/groups/ST/post-quantum-crypto/cfp-announce-dec2016.html},
  2016.
\newblock Accessed: 2017-04-02.

\bibitem[NS01]{NS01}
Phong~Q Nguyen and Jacques Stern.
\newblock The two faces of lattices in cryptology.
\newblock In {\em Cryptography and lattices}, pages 146--180. Springer, 2001.

\bibitem[NV08]{NguyenVidick08}
Phong~Q. Nguyen and Thomas Vidick.
\newblock Sieve algorithms for the {S}hortest {V}ector {P}roblem are practical.
\newblock {\em J. Math. Cryptol.}, 2(2):181--207, 2008.

\bibitem[Odl90]{Odl90}
Andrew~M Odlyzko.
\newblock The rise and fall of knapsack cryptosystems.
\newblock {\em Cryptology and computational number theory}, 42:75--88, 1990.

\bibitem[Pei08]{Peikert08}
Chris Peikert.
\newblock Limits on the hardness of lattice problems in $\ell_p$ norms.
\newblock {\em Computational Complexity}, 17(2):300--351, May 2008.
\newblock Preliminary version in CCC 2007.

\bibitem[Pei10]{Pei10}
Chris Peikert.
\newblock An efficient and parallel {G}aussian sampler for lattices.
\newblock In {\em {CRYPTO}}. 2010.

\bibitem[Pei16]{chris_survey}
Chris Peikert.
\newblock A decade of lattice cryptography.
\newblock {\em Foundations and Trends in Theoretical Computer Science},
  10(4):283--424, 2016.

\bibitem[PS09]{PS09}
Xavier Pujol and Damien Stehl{\'e}.
\newblock Solving the {S}hortest {L}attice {V}ector {P}roblem in time $2^{2.465
  n}$.
\newblock {\em IACR Cryptology ePrint Archive}, 2009:605, 2009.

\bibitem[Reg09]{oded05}
Oded Regev.
\newblock On lattices, learning with errors, random linear codes, and
  cryptography.
\newblock {\em Journal of the ACM}, 56(6):Art. 34, 40, 2009.

\bibitem[RR06]{RR06}
Oded Regev and Ricky Rosen.
\newblock Lattice problems and norm embeddings.
\newblock In {\em {STOC}}, 2006.

\bibitem[Sch87]{Schnorr87}
C.P. Schnorr.
\newblock A hierarchy of polynomial time lattice basis reduction algorithms.
\newblock {\em Theoretical Computer Science}, 53(23):201 -- 224, 1987.

\bibitem[Sha84]{Shamir84}
Adi Shamir.
\newblock A polynomial-time algorithm for breaking the basic {M}erkle-{H}ellman
  cryptosystem.
\newblock {\em IEEE Trans. Inform. Theory}, 30(5):699--704, 1984.

\bibitem[Ste16]{DGStoSVP}
Noah Stephens{-}Davidowitz.
\newblock Discrete {G}aussian sampling reduces to {CVP} and {SVP}.
\newblock In {\em SODA}, 2016.

\bibitem[{van}81]{Boas81}
Peter {van Emde Boas}.
\newblock Another {NP}-complete problem and the complexity of computing short
  vectors in a lattice.
\newblock Technical report, University of Amsterdam, Department of Mathematics,
  Netherlands, 1981.
\newblock Technical Report 8104.

\bibitem[WLTB11]{WLTB11}
Xiaoyun Wang, Mingjie Liu, Chengliang Tian, and Jingguo Bi.
\newblock Improved {N}guyen-{V}idick heuristic sieve algorithm for shortest
  vector problem.
\newblock In {\em ASIACCS}, 2011.

\end{thebibliography}
\newcommand{\etalchar}[1]{$^{#1}$}
\def\cprime{$'$}

\end{document}